\documentclass[aps,prx,twocolumn,showpacs,superscriptaddress,preprintnumbers,amssymb,floatfix,10pt]{revtex4-2} 
\usepackage{graphicx}
\usepackage{latexsym}
\usepackage{amsmath}
\usepackage{mathtools}
\usepackage{amsfonts}
\usepackage{rotating}
\usepackage{upgreek}
\usepackage{bm}
\usepackage{multirow}
\usepackage[shortlabels]{enumitem}
\usepackage{xcolor}
\usepackage{physics}
\usepackage{tcolorbox}
\usepackage{manfnt}
\usepackage{lipsum}
\usepackage[ruled,linesnumbered]{algorithm2e}
\usepackage{tabularx}
\usepackage{nicematrix}
\usepackage{braket}
\usepackage[normalem]{ulem}
\usepackage{tikz}
\usepackage{wrapfig}
\usepackage{setspace}
\usepackage{stmaryrd}
\usepackage{amsthm}
\usepackage[colorlinks, citecolor=blue]{hyperref}
\usepackage{orcidlink}

\newcommand{\secref}[1]{Sec.\,\ref{#1}}

\newcommand{\eqnref}[1]{Eq.\,\eqref{#1}}

\newcommand{\figref}[1]{Fig.\,\ref{#1}}

\newtheorem{definition}{Definition}
\newtheorem{lemma}{Lemma}

\newtheorem{theorem}{Theorem}
\newtheorem{corollary}{Corollary}

\newtheorem{result}{Result}

\makeatletter
\newcommand{\cyclic}[1]{\ensuremath{\cyclic@i#1,\@nil}}
\def\cyclic@i#1,#2\@nil{%
  #1%
  \ifx\relax#2\relax
  \else
    \!\shortrightarrow\!\cyclic@i#2\@nil
  \fi
}
\makeatother

\begin{document}

\title{Liouvillian Gap in Dissipative Haar-Doped Clifford Circuits}  

\author{Ha Eum Kim~\orcidlink{0009-0004-2614-0834}}
\email{Contact author: haeumkim@illinois.edu}
\affiliation{Department of Physics, University of Illinois at Urbana-Champaign, Urbana, Illinois 61801, USA}
\affiliation{Department of Electrical and Computer Engineering, University of Illinois at Urbana-Champaign, Urbana, IL 61801, USA}
\affiliation{Department of Mathematics and Research Institute for Basic Sciences, Kyung Hee University, Seoul, 02447, Korea}

\author{Andrew D. Kim~\orcidlink{0009-0003-9966-7365}} \email{Contact author: adk7@illinois.edu}
\affiliation{Department of Physics, University of Illinois at Urbana-Champaign, Urbana, Illinois 61801, USA}

\author{Jong Yeon Lee~\orcidlink{0000-0002-7387-3326}}
\email{Contact author: jongyeon@illinois.edu}
\affiliation{Department of Physics, University of Illinois at Urbana-Champaign, Urbana, Illinois 61801, USA}
\affiliation{Korea Institute for Advanced Study, Seoul 02455, South Korea}

\date{\today}% It is always \today, today,

\begin{abstract}

Quantum chaos is commonly assessed through probe-dependent signatures that need not coincide.
Recently, a dissipative signature was proposed for chaotic Floquet systems, where infinitesimal bulk dissipation induces a non-zero constant intrinsic relaxation rate quantified by the Liouvillian gap.
This raises a question: what minimal departure from Clifford dynamics is required to generate such intrinsic relaxation? To address this, we study a Floquet two-qubit Clifford circuit doped with Haar-random single-qubit gates and subject to local dissipation of strength $\gamma$.
We find a structure-dependent crossover.
The undoped \textit{i}SWAP-class circuit exhibits a weak-dissipation singularity, with a gap that grows with $N$ for any $\gamma>0$. 
Haar doping preserves this undoped-like growth for any subextensive doping pattern.
At finite doping density, there exist patterns that yield an $\mathcal{O}(1)$ gap for any fixed $\gamma$ as $N\to\infty$, yet remain singular as $\gamma\to0^+$.
Because our bounds depend only on the spatial doping pattern, they remain valid even when the Haar rotations are independently redrawn each Floquet period.
Overall, our findings provide a circuit-level perspective on intrinsic relaxation—and thus irreversibility—in open many-body systems.
\end{abstract}

\maketitle

% \tableofcontents

\begin{figure*}[!t]
    \centering
    \includegraphics[width=0.98\linewidth]{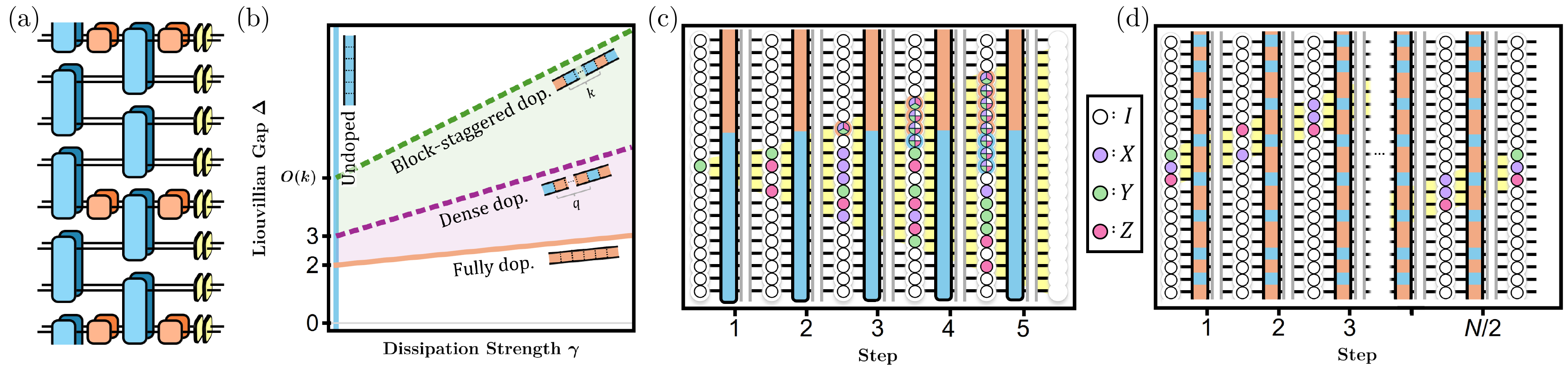}
    \caption{ {\bf Summary of analytic results.}
    Dissipative Haar-doped Floquet circuit and the weight-truncation framework for bounding the Liouvillian gap.
    (a) One Floquet period consists of a two-layer Clifford brickwork built from a fixed \textit{i}SWAP-class gate on odd and even bonds (blue), with single-qubit Haar rotations applied immediately after each layer on a chosen set of $n_h$ qubit lines (orange), followed by an onsite depolarizing layer on all qubits (yellow).
    In (b)-(d), each long rectangle represents the qubit lines, with blue and orange segments indicating undoped and doped sites, respectively.
    (b) Diagram of Liouvillian gap $\Delta$ versus dissipation strength $\gamma$ for representative doping patterns at the thermodynamic limit.
    Solid lines show exact results for fully doped (orange); dashed lines show analytic upper bounds for dense doping (light purple) and block-staggered doping with spacing $k$ between doped (green). Unlike the doped patterns, whose singularities remain independent of system size, the undoped circuit (blue) exhibits a singularity that diverges as the system size $N$ increases.
    (c,d) Lower- and upper-bound constructions. Gray capacitors denote Pauli-weight truncation projectors inserted between Floquet steps.
    (c) Representative truncation of an initially local $Y$ operator (green) once its weight exceeds $w_{\mathrm{t}}=6$ and vanishes after five steps.
    (d) Under truncation, return cycles can only be supported by Pauli strings confined to a small consecutive block, illustrated by a three-site string $YXZ$ that translates by two sites per step and returns after $N/2$ steps.
    }
    \label{fig:circuit}
\end{figure*}

\section{\label{sec:intro}Introduction}

In classical mechanics, chaos is defined at the level of trajectories based on the sensitivity to initial conditions.
In quantum many-body systems, no single definition plays the same role; quantum chaos is typically inferred from a collection of signatures, including spectral statistics~\cite{haake1991quantum,bohigas1984characterization,guhr1998random,d2016quantum,deutsch2018eigenstate,eisert2015quantum,leblond2021universality}, out-of-time-order correlation (OTOC) based measures of operator growth/scrambling~\cite{roberts2017chaos,harrow2021separation,hosur2016chaos,nahum2018operator}, or entanglement growth diagnostics~\cite{prosen2007efficiency,nahum2017quantum}.
However, these signatures can disagree~\cite{dowling2023scrambling,von2018operator,turner2018weak,pizzi2025genuine}, so no single probe is universally definitive.

The emergence of irreversibility from pure Hamiltonian dynamics is often discussed using concepts from quantum chaos, in particular the eigenstate thermalization hypothesis (ETH)~\cite{deutsch1991quantum,srednicki1994chaos,rigol2008thermalization,mori2018thermalization}.
Beyond this unitary perspective, recent work proposes an \emph{intrinsic} relaxation signature for chaotic Floquet systems~\cite{mori2024liouvillian}. 
Specifically, for a Lindbladian dynamics with dissipation strength $\gamma$, one takes system size $N\to\infty$ first and then $\gamma\to0^+$, and asks whether the \emph{Liouvillian gap} remains finite.
This weak-dissipation perspective has since been developed and extended in a variety of settings~\cite{yoshimura2024robustness,yoshimura2025theory,zhang2024thermalization,jacoby2025spectral}.
Separately, spectral form factor diagnostics have been used to probe many-body chaos through controlled nonunitary deformations~\cite{liao2022emergence}.

This signature is tied to operator spreading: without conservation laws or other hydrodynamic slow modes, rapid spreading makes local dissipation effectively act on an extensive support, yielding a finite intrinsic gap.
By contrast, integrability or (quasi-)conservation can suppress this weak-dissipation response or alter it qualitatively~\cite{garcia2023keldysh,vznidarivc2024momentum,duh2025ruelle}.

Clifford circuits occupy a singular corner of this landscape. By familiar closed-system diagnostics, most notably spectral statistics, they do not exhibit random-matrix universality, which is often taken as a hallmark of generic quantum chaos. At the same time, Clifford dynamics can display maximal spreading of local Pauli operators. 
More strikingly for our purposes, we show that it can yield an extensive Liouvillian (channel) gap under infinitesimal local dissipation. 
At first glance, this is counterintuitive: a finite intrinsic relaxation rate is often interpreted as evidence of strong many-body mixing, whereas Clifford circuits lack random-matrix universality and remain efficiently classically simulable in the stabilizer formalism, and are therefore usually regarded as far from chaotic.

Motivated by this tension, we start from a Clifford backbone and introduce non-Clifford doping as a controlled departure. Doping injects non-stabilizer resources, allowing one to interpolate between a stabilizer circuit ensemble and a Haar-random one. In many works, this interpolation is tracked through spectral and eigenstate diagnostics, which cross over toward random-matrix universality; importantly, the doping scale required is probe dependent. For instance, entanglement spectral statistics can cross over at ${\cal O}(1)$ doping~\cite{zhou2020single}, whereas higher-point OTOC-based diagnostics may require resources scaling extensively with system size~\cite{leone2021quantum}. 
Instead, in this work, we apply the intrinsic gap diagnostic to understand dissipative Haar-doped Floquet Clifford circuits, illustrated in Fig.~\ref{fig:circuit}(a) and detailed in Sec.~\ref{sec:Model}. 
We summarize our main results below.

\subsection{Summary of Results}

We begin with the undoped Clifford limit as a baseline.
In \secref{sec:non-doped}, we find that the undoped \textit{i}SWAP-class Floquet circuit has a Liouvillian gap that scales as $\Delta\propto N$ for any dissipation strength $\gamma$, whereas other locally Clifford-equivalent two-qubit gate classes do not.
For concreteness, we fix a representative \textit{i}SWAP-class Clifford gate throughout.
Once we introduce Haar doping, the dynamics is no longer purely Clifford, so the Liouvillian gap is not exactly tractable at finite $\gamma$ by stabilizer methods alone.
In \secref{sec:weak_dissipation}, we use numerics to probe the weak-dissipation regime and find the Haar-averaged gap scales as $\langle \Delta\rangle \sim N$ (see Fig.~\ref{figs:gap_scaling}(a)--(c)).
This scaling indicates that the weak-dissipation response persists under Haar doping.

To understand the crossover in system-size scaling away from the $\gamma\to0^+$ limit, we turn to the analytically tractable strong-dissipation regime.
In this regime, the effective Pauli dynamics admits a natural description in terms of a weight truncation with a cutoff set by the dissipation strength.
In \secref{sec:low_bound}, we first obtain the following lower bound on the Liouvillian gap:
\begin{result}[Strong-dissipation lower bound (informal)]
In the strong-dissipation regime $\gamma \gg 1$, for $n_h\ge 1$ Haar-doped sites placed at arbitrary positions, the Liouvillian gap satisfies
\begin{align}
\Delta \gtrsim \frac{N}{n_h}\,\gamma,
\end{align}
at the level of scaling.
\label{res:strong_dissipation_lower_bound}
\end{result}
\noindent
This bound can be viewed as a minimal-cutoff condition: if all trajectories initiated below the cutoff are eventually driven above it by Pauli-weight growth through undoped regions (see \figref{fig:circuit}(c)), then the Liouvillian gap cannot be captured within the truncated sector alone.
A precise statement and proof are given in Theorem~\ref{thm:gen_gap-ncl}.
This lower bound captures the qualitative role of doping density.
If $n_h/N\to 0$ as $N\to\infty$, it forces the gap to grow with system size, whereas at finite doping density it does not.
In \secref{sec:struc_dop}, we show that at nonvanishing doping density the Liouvillian gap can remain finite for suitable doping patterns as depicted in \figref{fig:circuit}(b).

\begin{figure*}[!t]
    \centering
    \includegraphics[width=0.98\linewidth]{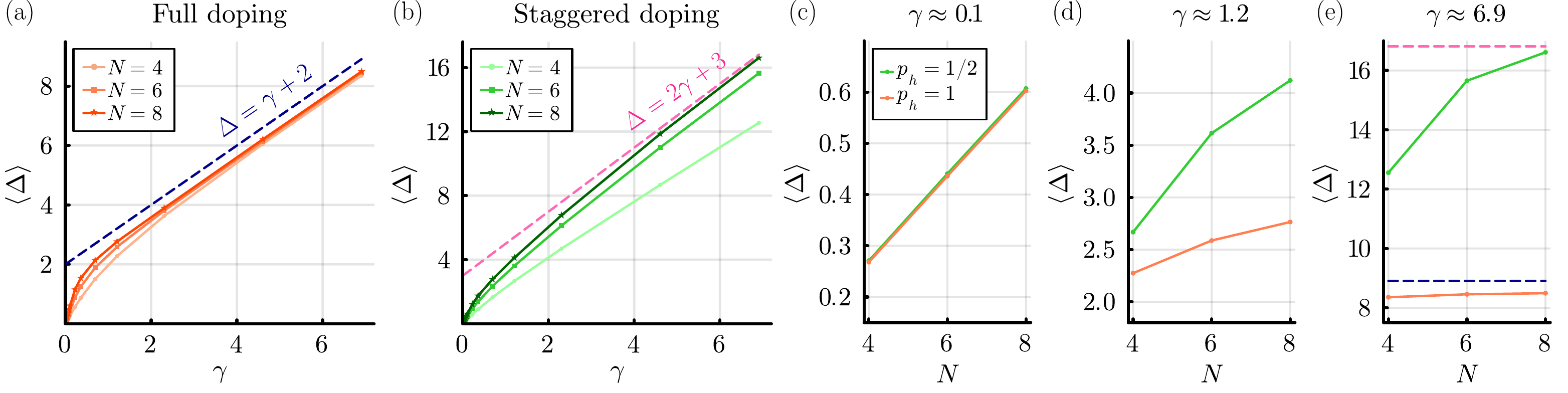}
    \caption{{\bf Summary of numerical results.} Ensemble-averaged Liouvillian gap $\langle\Delta\rangle$ of the DHFC as a function of the dissipation strength $\gamma$ at (a) full and (b) staggered doping for $N\in\{4,6,8\}$.
    Dashed lines (with accompanying formulae) give the analytic results for the gap in the thermodynamic and strong dissipation limits. 
    As $N$ increases, the discrepancy between $\langle \Delta \rangle$ and the analytic result decreases; the finite-$N$ scaling of our numerically obtained $\langle \Delta \rangle$ is shown at (c) weak, (d) intermediate, and (e) strong dissipation strengths $\gamma$. 
    }
    \label{figs:gap_scaling}
\end{figure*}

As an opposite extreme to the undoped circuit, the fully doped pattern provides a natural reference point with maximal local Haar mixing, for which the gap is obtained as $\Delta=\gamma+2$ (Theorem~\ref{thm:ful_lgap}).
For dense doping (no adjacent undoped sites), we derive an explicit upper bound $\Delta\le 3\gamma+3$ (Theorem~\ref{thm:dense}), while for block-staggered doping with $k$ consecutive undoped sites between doped sites we obtain $\Delta\le A\gamma+B$ with $A,B=O(k)$, implying an $N$-independent gap for any fixed $k$ (see \secref{subsec:block_dop}).
These upper bounds follow from how trajectories behave under the truncated effective Pauli dynamics: ballistically spreading strings inevitably traverse an undoped region, where deterministic Clifford growth pushes the Pauli weight above the cutoff and the path is removed by truncation.
By contrast, the surviving contributions are return cycles confined to a small contiguous block that evolve essentially by directed translation and close after $N/2$ steps (see \figref{fig:circuit}(d)).

Taken together, our results identify two thermodynamic regimes for the Liouvillian gap: undoped-like system-size growth versus a fully-doped-like $O(1)$ gap.
If $n_h/N\to 0$, the lower bound forces growth for any pattern, whereas at nonvanishing doping density a crossover becomes possible and suitable structures can sustain a finite gap as $N\to\infty$.
Because our bounds depend only on the spatial doping pattern, the same picture applies both to quenched Haar disorder and to Haar rotations resampled independently each Floquet period.

\section{Dissipative Haar-Doped Circuit}
\label{sec:Model}

We begin by specifying the dissipative Haar-doped circuit that will serve as our main model and by introducing the Liouvillian gap as our diagnostic of relaxation.
In \secref{sec:notation}, we introduce the basic notation used through the manuscript. 
\secref{subsec:floq_gap} defines the one-period Floquet channel and Liouvillian gap, and \secref{subsec:gamma} clarifies our convention that the dissipation strength $\gamma$ is defined independently of the system size $N$.
\secref{subsec:pdynamic_trunc} then introduces a weight-truncated Pauli description and uses it to obtain analytic upper and lower bounds.
This framework is not tied to quenched doping and can be adapted to the case where the doped single-qubit rotations are independently redrawn each Floquet period.

\subsection{Notation}
\label{sec:notation}

We work primarily in the Pauli basis and write $\{\sigma_\mu\}_{\mu=0}^3$ with $\sigma_0=I$ and $\sigma_{1,2,3}=(X,Y,Z)$.
The operator acting on site $j$ is denoted by $\sigma_\mu^j$, and we freely interchange $\sigma_{1,2,3}$ with $X,Y,Z$ when convenient.
When writing Pauli strings, we omit explicit tensor products and, when site labels are unambiguous, also drop position indices; for example, on sites $(j,j{+}1,j{+}2)$ we write $XYZ$ (or $X_jY_{j+1}Z_{j+2}$) to mean $X_j\otimes Y_{j+1}\otimes Z_{j+2}$, with identities on all other sites implied.

We also adopt the following conventions for Pauli strings, their support, and their weight.
A Pauli string on $N$ qubits is an operator of the form
\begin{align}
\label{eq:pstring}
S := \bigotimes_{j=1}^{N} \sigma_{\mu_j}
, \qquad \sigma_{\mu_j}\in\{I,X,Y,Z\},
\end{align}
where we identify strings that differ only by a global phase in $\{\pm1,\pm i\}$.
The set of such strings forms the Pauli basis, denoted \(\mathcal{P}_N\).
The support of $S$ is the set of qubits on which it acts nontrivially:
\begin{align}
\mathrm{supp}(S) := \{\, j\in\{1,\dots,N\} : \sigma_{\mu_j}\neq I \,\},
\end{align}
and the Pauli weight is the size of the support:
\begin{align}
\label{def:weight}
w(S) := |\mathrm{supp}(S)|.
\end{align}

When describing doping structures in text, we use the symbols \(\circ\) and \(\bullet\) to denote undoped and doped qubit lines, respectively.

\subsection{Floquet Channel and Liouvillian Gap}
\label{subsec:floq_gap}

We consider a one-dimensional periodic chain of \(N\) qubits~\footnote{
Boundary conditions play no substantive role in the main discussion; for clarity in the numerics and auxiliary lemmas, we use periodic boundary conditions.}.
We consider a Floquet dynamics consisting of a unitary stage followed by dissipation, see  \figref{fig:circuit} (a).
The unitary stage is a two-layer brickwork built from a fixed two-qubit Clifford gate,
\begin{align}
\label{eq:iSWAP}
    u:=i\mathrm{SWAP}\,(H \otimes H),
\end{align}
with $i\mathrm{SWAP}:=\exp\!\left[i\pi/4(X\otimes X+Y\otimes Y)\right]$.
The two layers act on odd and even bonds, respectively.
Single-qubit Haar rotations are inserted immediately after the two-qubit Clifford gate. 
We refer to theses single-qubit Haar rotations as \emph{Haar doping}.
Concretely, within each Floquet period we apply a Haar random gate on a chosen site immediately after each of the two Clifford layers, so that the same qubit line is doped once after each layer.
We denote by $n_{\mathrm{\mathrm{h}}}$ the number of doped qubit lines and by $n_{\mathrm{cl}}:=N-n_h$ the number of undoped (Clifford-only) qubit lines.
The doping density is defined as
\begin{equation}
    p_{\mathrm{h}} := \frac{n_{\mathrm{h}}}{N}.
\end{equation}
As the final step of each period we apply a site-local Pauli channel across the chain to implement dissipation. We refer to this circuit as the dissipative Haar-doped Floquet circuit (DHFC).

One feature of our circuit is that we fix the two-qubit Clifford gate to the unitary $u$ defined in Eq.~\eqref{eq:iSWAP}, which belongs to the \textit{i}SWAP class~\footnote{Ref.~\cite{grier2022classification} denotes this same local-Clifford class by FSWAP.}, one of four local-Clifford equivalence classes~\cite{grier2022classification}.
We make this choice to avoid low-depth localization artifacts that can arise in shallow two-layer brickworks, since an irreducible localization structure requires CZ class gates from the two layers to intersect so as to form the wall boundaries~\cite{kovacs2024operator}.
The \textit{i}SWAP class is instead known to exhibit ballistic spreading~\cite{bertini2019exact} and to provide fast convergence to a unitary $2$-design in gadget constructions~\cite{kong2024convergence}.
Accordingly, in Sec.~\ref{sec:non-doped} we analyze the undoped dynamics in detail as a function of the underlying Clifford gate class and provide evidence that \textit{i}SWAP-class circuits exhibit a cyclic-averaged Pauli weight that grows proportionally with the system size.

Our circuit models an experimental setting in which the system is weakly coupled to a Markovian environment, so that dissipation is well described by a local Lindblad generator.
During a coherent stage of duration $\tau_U$, the brickwork drive is applied on a timescale fast compared to the bath coupling, so dissipation is negligible and the evolution is effectively unitary.
During a subsequent dissipative stage of duration $\tau_D$, the drive is strongly detuned and the dynamics are governed by a local Lindblad generator.
We formalize this piecewise dynamics using a time-dependent Lindblad master equation~\cite{nielsen2010quantum},
\begin{equation}
\frac{d\rho}{dt}=\mathcal{L}(t)[\rho]
= -\,i[H(t),\rho]
+ \sum_{j=1}^{N}\sum_{\mu=1}^{3}\mathcal{D}[L_{j,\mu}(t)](\rho),
\label{eq:lindblad}
\end{equation}
with Floquet period $\tau=\tau_U+\tau_D$ and periodic continuation $\mathcal{L}(t+\tau)=\mathcal{L}(t)$.
Here $H(t)$ is piecewise constant,
\begin{align}
H(t)=
\begin{cases}
H_{\mathrm{drive}}, & 0\le t<\tau_U,\\
0, & \tau_U\le t<\tau,
\end{cases}
\end{align}
so that the propagator over $[0,\tau_U)$,
$U=\mathcal{T}\exp\!\big(-i\!\int_0^{\tau_U}\!dt\,H_{\mathrm{drive}}\big)$,
implements the coherent stage of the circuit.
The dissipator is $\mathcal{D}[L](\rho):=L\rho L^\dagger-\tfrac12\{L^\dagger L,\rho\}$, and we take the jump operators to be
\begin{align}
L_{j,\mu}(t)=
\begin{cases}
0, & 0\le t<\tau_U,\\[2pt]
\sqrt{\Gamma_{j,\mu}}\,\sigma^j_{\mu}, & \tau_U\le t<\tau,
\end{cases}
\qquad \Gamma_{j,\mu}\ge 0.
\end{align}
To avoid site- or axis-dependent anisotropies, we specialize to the homogeneous isotropic case $\Gamma_{j,\mu}\equiv\gamma_0$.
Over the dissipative subinterval of length $\tau_D$, the evolution reduces to
$\exp\!\big[\tau_D\,\mathcal{L}_{\mathrm{dis}}\big]$ with
$\mathcal{L}_{\mathrm{dis}}:=\gamma_0\sum_{j=1}^{N}\sum_{\mu=1}^{3}\mathcal{D}[\sigma^j_{\mu}]$,
which corresponds to an onsite depolarizing channel.
Since the resulting map depends only on $\tau_D\gamma_0$, we define the dimensionless dissipation strength $\gamma:=\tau_D\gamma_0$.

% Liouvillian gap
We introduce the one-cycle Floquet channel
\begin{align}
\label{eq:one-floq}
\Phi_F\;:=\;\mathcal{T}\exp\!\Big[\int_{0}^{\tau}\!\mathcal{L}(t)\,dt\Big].
\end{align}
Let $e^{\lambda_a\tau}$ denote the eigenvalues of $\Phi_F$, with $\lambda_0=0$ corresponding to the unique periodic steady state.
Uniqueness follows from the depolarizing stage, which mixes any initial state toward the maximally mixed state and hence renders the one-period channel primitive.
We refer to $\{\lambda_a\}$ as the Liouvillian eigenvalues in the Floquet setting.
The Liouvillian gap is defined as
\begin{align}
\Delta\;:=\;-\max_{a\neq 0}\mathrm{Re}\,\lambda_a,
\end{align}
i.e., the smallest per-period decay rate among all non-stationary modes.
Our central observable is the Liouvillian gap $\Delta$, which captures the relaxation scale of DHFC.
Throughout, the thermodynamic-limit behavior of $\Delta$ does not require an explicit Haar average over the single-qubit rotations on the doped sites.
As noted in Sec.~\ref{sec:struc_dop}, the relevant quantities are self-averaging for a typical circuit realization.
Intuitively, this is because the dependence on the doped single-qubit rotations enters through products of many independent Haar-random factors, which concentrate in the large-$N$ limit.
Nevertheless, for finite system sizes, as in Sec.~\ref{sec:weak_dissipation}, we present numerical data based on the Haar-averaged Liouvillian gap,
\begin{align}
\langle \Delta\rangle
\;:=\;\mathbb{E}_{\mathrm{Haar}}[\Delta],
\end{align}
which suppresses fluctuations across different Haar realizations, i.e., across independent choices of the single-qubit rotations on the doped sites, and improves the stability of the scaling analysis.

\subsection{Size-Independent Dissipation Strength}
\label{subsec:gamma}

As emphasized by Mori~\cite{mori2024liouvillian}, the weak-dissipation limit in chaotic open Floquet systems can be singular due to the interplay between operator spreading under the unitary dynamics and bulk dissipation.
Because dissipation accumulates over operator support, sufficiently weak local noise can generate a finite intrinsic relaxation rate once operators become system-spanning.
Consequently, while the Liouvillian gap $\Delta$ trivially closes as $\gamma\to 0$ at fixed $N$, taking the thermodynamic limit first can yield a qualitatively different outcome, so that the two limits need not commute:
\begin{align}
\lim_{N\to\infty}\lim_{\gamma\to 0}\Delta
\neq
\lim_{\gamma\to 0}\lim_{N\to\infty}\Delta .
\label{eq:noncomm_limits}
\end{align}

To interpret Eq.~\eqref{eq:noncomm_limits}, it is important to specify how the dissipation strength $\gamma$ is defined in our model.
Here we interpret the single-site depolarizing channel not as an effective coarse-grained description whose strength could in principle depend on system size, but as arising from local environments whose fluctuations are uncorrelated across sites.
Accordingly, $\gamma$ is fixed by the local system--bath coupling during a Floquet period and is specified independently of the total system size $N$.

\subsection{Effective Pauli Dynamics via Weight Truncation}
\label{subsec:pdynamic_trunc}

The one-period Floquet channel defined in Eq.~\eqref{eq:one-floq} can be written as
\begin{align}
\Phi_F := \mathcal{N}\circ \mathcal{U},
\label{eq:Floq_channel}
\end{align}
where $\mathcal{U}(\rho):=U\rho U^\dagger$ denotes the unitary circuit channel and
$\mathcal{N}$ is an on-site depolarizing layer acting independently on each qubit,
with dissipation strength $\gamma$.
We work in the Liouville-space representation of these CPTP maps.
In the Pauli-string basis $\mathcal{P}_N$, the noise layer is diagonal:
\begin{align}
\mathcal{N}(S)=e^{-\gamma\, w(S)}S,\qquad S\in\mathcal{P}_N,
\label{eq:depol_diag}
\end{align}
where $w(S)$ is the Pauli weight (see Appendix~\ref{app:dep}).

Equation~\eqref{eq:depol_diag} shows that the depolarizing layer suppresses Pauli strings exponentially in their weight, so that large-weight components are strongly attenuated and the dynamics is effectively governed by a small-weight sector of Liouville space.
This motivates a truncated description in which we retain only Pauli strings up to a cutoff weight $w_{\mathrm{t}}$.
To implement this truncation in a symmetric form around the unitary step, it is convenient to work with the symmetrized Floquet channel
\begin{equation}
\tilde{\Phi}_F := \mathcal{N}^{1/2}\circ \mathcal{U}\circ \mathcal{N}^{1/2},
\label{eq:sym_Floq_channel}
\end{equation}
which is related to $\Phi_F$ by a similarity transform and therefore shares the same nonzero spectrum.
Letting $\Pi_{w_{\mathrm{t}}}$ denote the projector onto $\mathrm{span}\{P\in\mathcal{P}_N:\, w(P)\le w_{\mathrm{t}}\}$, we define the truncated map
\begin{equation}
\tilde{\Phi}^{\mathrm{trunc}}_{w_{\mathrm{t}}}
:= \Pi_{w_{\mathrm{t}}}\,\tilde{\Phi}_F\,\Pi_{w_{\mathrm{t}}}.
\label{eq:trunc_Floq_channel}
\end{equation}

With $w_{\mathrm{t}}$ determined by $\gamma$, weight truncation provides a tractable framework for deriving upper and lower bounds on the Liouvillian gap in the strong dissipation regime.
For the lower bound, we use that propagation through the undoped region tends to drive operators out of the low-weight sector.
Under the fixed brickwork generated by $u$, Pauli strings propagate without branching, and their weight typically grows past the cutoff $w_{\mathrm{t}}$.
By contrast, at doped sites the single-qubit rotations mix the local Pauli basis and split an operator into multiple Pauli components, among which some can remain within the low-weight sector.
This approach is illustrated in \figref{fig:circuit} (c) where, in a contiguous circuit with a doped block (orange) and an undoped block (blue), an initially local $Y$ operator (green) spreads under successive Floquet steps and is truncated once its weight exceeds $w_{\mathrm{t}}=6$.
In particular, suppose there exists a finite $k$ such that every Pauli string in $\mathcal{P}_{\le w_{\mathrm{t}}}$ is driven above the cutoff under finite $k$ iterations.
Here $\mathcal{P}_{\le w_{\mathrm{t}}}$ denotes the set of $N$-qubit Pauli strings with weight at most $w_{\mathrm{t}}$.
Then the truncated propagator is nilpotent,
\begin{equation}
\big(\tilde{\Phi}^{\mathrm{trunc}}_{w_{\mathrm{t}}}\big)^k = 0,
\label{eq:nilpotent_trunc}
\end{equation}
so the low-weight sector supports no nontrivial eigenvalues.
Consequently, any slow decay of the full channel must involve excursions into higher-weight sectors, yielding a lower bound on the Liouvillian gap in the strong dissipation regime.
A precise formulation and proof are given in Sec.~\ref{sec:low_bound}.

We now outline the strategy behind our upper bound, with the detailed derivation for the present circuit deferred to Sec.~\ref{sec:struc_dop}.
The key step is to determine the minimal cutoff $w_{\mathrm{trunc}}$ for which the truncated Floquet dynamics $\tilde{\Phi}^{\mathrm{trunc}}_{w_{\mathrm{t}}}$ becomes non-nilpotent, meaning that the truncated Pauli dynamics supports a return cycle.
At this cutoff, the unitary step together with the projection induces a directed support graph on Pauli strings in $\mathcal{P}_{\le w_{\mathrm{trunc}}}$, and the overlap
$\Tr\!\left[P\,(\tilde{\Phi}^{\mathrm{trunc}}_{w_{\mathrm{t}}})^L(P)\right]$
is obtained by summing the contributions of all $L$-step return paths from the Pauli string $P$ back to itself.
To identify which paths can contribute, we use the constraint imposed by truncation.
Once a Pauli string spreads ballistically, its support inevitably enters the undoped region, where the deterministic Clifford propagation drives rapid weight growth and the component is removed by truncation.
Consequently, the only surviving paths are confined to a small block of consecutive sites and evolve essentially by translating in a fixed direction.
\figref{fig:circuit} (d) illustrates this picture for a circuit with an undoped island structure where a Pauli string supported on three consecutive sites translates upward by two sites per truncated Floquet step and returns to the same Pauli string at its original location after $N/2(=L)$ steps.

When such return paths repeat with the same structure under successive traversals, their contributions can accumulate coherently across iterations.
Gelfand’s formula then relates the resulting long-time growth to the nontrivial spectral radius of the truncated propagator,
\begin{equation}
\rho\!\left(\tilde{\Phi}^{\mathrm{trunc}}_{w_{\mathrm{t}}}\right)
\ge
\left|\Tr\!\left[P\,(\tilde{\Phi}^{\mathrm{trunc}}_{w_{\mathrm{t}}})^L(P)\right]\right|^{1/L}.
\label{eq:spec_low_bound}
\end{equation}
Since the gap is the slowest per-period decay rate, the bound in Eq.~\eqref{eq:spec_low_bound} directly translates into an upper bound on the gap.
We evaluate this upper-bound estimate across representative doping structures, and \figref{fig:circuit}(b) summarizes the resulting Liouvillian gap $\Delta$ as a function of the dissipation strength $\gamma$ at the thermodynamic limit $N\to\infty$.
For the fully doped circuit (orange solid line) and the staggered pattern in which doped and undoped sites alternate, the minimal cutoff weights are $w_t=1$ and $w_t=2$, respectively, so that a unique return cycle exists within the truncated dynamics and yields the exact gaps $\Delta=\gamma+2$ and $\Delta=2\gamma+3$.

For more general patterns, our construction yields explicit upper bounds.
These are shown in \figref{fig:circuit}(b) as the light-purple dashed line for dense doping, where every undoped site is separated by at least one doped site, and as the green dashed line for block-staggered doping, where doped sites are separated by $k$ undoped sites.
For dense doping, we find $\Delta\le 3\gamma+c$ with $c=O(1)$, emphasizing that the singularity is similar in form to the fully doped case in that in both cases the gap does not grow with system size.
For block-staggered doping, we obtain a bound of the form $\Delta\le A\gamma+B$ where the dissipative slope $A$ scales at
most linearly with $k$.
Since the doping density scales as $p_h=1/(k+1)$, this result suggests that a system-size-independent gap can persist even when the doping becomes parametrically sparse.
Taken together, the doped structures above yield strong dissipation singularities that remain finite in the thermodynamic limit, whereas the undoped circuit shown as the blue solid line has $\Delta=(N/2)\gamma$ and hence diverges in the thermodynamic limit as $N\to\infty$.

Note that the upper and lower bounds derived above are determined by the spatial doping pattern of the one-period Floquet channel.
They follow from how this pattern constrains the coupling between Pauli weight sectors under the unitary layer followed by weight truncation, rather than from reusing the same local Haar rotations each period.
Accordingly, the same bounds apply both for quenched Haar disorder and for Haar rotations resampled independently at each Floquet step.
From this perspective, our bounds depend on how the doping pattern affects the Pauli spreading dynamics generated by the Clifford backbone.

\section{Non-doped Clifford Circuit}
\label{sec:non-doped}

We begin with the non-doped dissipative Clifford Floquet circuit as a solvable reference point among the doping structures considered in this work.
In this setting the Liouvillian gap can be obtained throughout the full dissipation-strength regime without invoking the truncation method of Sec.~\ref{subsec:pdynamic_trunc}, since the Clifford layer acts as a signed permutation on Pauli strings.
In Sec.~\ref{subsec:non_doped_eigen} we show that the Floquet spectrum decomposes into Clifford orbits, with eigenvalues determined by the orbit-averaged Pauli weight.
In Sec.~\ref{subsec:abs_local} we then examine how the choice of the underlying Clifford gate affects the presence or absence of a localized spectral sector in this shallow-depth setting.
For \textit{i}SWAP-class circuits, we find no deterministic localized sector, which leads to a gap that grows linearly with the system size.

\subsection{Eigenspectrum}
\label{subsec:non_doped_eigen}

In this subsection, we show that throughout the full dissipation regime, the eigenspectrum of the non-doped Floquet channel $\Phi_F$ is proportional to the orbit-averaged Pauli weight $\bar w(S)$.
For the non-doped setting described by Eq.~\eqref{eq:Floq_channel}, each Floquet cycle consists of a unitary evolution generated by a Clifford circuit, followed by a dissipative interval governed by local depolarization. The resulting one-period channel takes the form
\begin{align}
\Phi_F := \mathcal{N} \circ \mathcal{C},
\end{align}
where \(\mathcal{C}\) is the Clifford circuit channel.

The action of the Clifford channel \(\mathcal{C}\) is particularly simple in the Pauli basis. Since Clifford unitaries conjugate Pauli strings to Pauli strings up to a sign, the channel acts as a permutation on \(\mathcal{P}_N\). Specifically, for any \(S \in \mathcal{P}_N\), we have
\begin{align}
\mathcal{C}(S) = CSC^\dagger = \pm S',
\end{align}
where \(S' \in \mathcal{P}_N\) and \(C\) is the unitary circuit generating \(\mathcal{C}\). As \(\mathcal{P}_N\) is finite, repeated application of \(\mathcal{C}\) generates a closed orbit. That is, for each \(S_0 \in \mathcal{P}_N\), there exists a minimal integer \(L\) such that
\begin{align}
\mathcal{C}^L(S_0) = \pm S_0.
\end{align}
This structure partitions the Pauli basis into disjoint orbits, within which the action of \(\mathcal{C}\) amounts to a permutation. Accordingly, \(\mathcal{C}\) admits a block-diagonal form in the Pauli basis, with each block corresponding to a distinct orbit.
In contrast to the permutation action of the Clifford channel, the depolarizing noise channel $\mathcal{N}$ acts diagonally in the Pauli basis. 
As derived in Appendix~\ref{app:dep}, each Pauli string $S$ is simply rescaled as $\mathcal{N}(S)=\exp[-2\gamma w(S)]S$.

Building on the Pauli-string basis representation of the Clifford and depolarizing channels, the eigenspectrum of the Floquet map \(\Phi_F\) can be directly characterized.
Each Pauli string \(S\) undergoes a closed orbit of length \(L\) under repeated Clifford conjugation, with its amplitude damped at each step according to the Pauli weight of the intermediate string.
After completing the orbit, we find
\begin{align}
    \Phi_F^L[S] 
    &= (\mathcal{N} \circ \mathcal{C})^L [S] \nonumber \\
    &= \left[ \prod_{j=0}^{L-1} e^{-\gamma w(\mathcal{C}^j[S])} \right] S.
\end{align}
This relation shows that \(S\) is an eigenoperator of \(\Phi_F^L\), and hence contributes to the spectrum of \(\Phi_F\) itself.
The corresponding Liouvillian eigenvalue has real part
\begin{align}
    -\mathrm{Re}[\lambda_S] = \gamma\bar{w}(S),
    \label{eq:eign_mag}
\end{align}
where the averaged Pauli weight over the orbit is defined by
\begin{align}
    \bar{w}(S) := \frac{1}{L} \sum_{j=0}^{L-1} w(\mathcal{C}^j[S]).
    \label{def:avg_weight}
\end{align}
The eigenoperator associated with \(\mathrm{Re}[\lambda_S]\) is supported entirely within the orbit generated by \(S\). The degeneracy of \(\mathrm{Re}[\lambda_S]\) is at least the size of the orbit, although further degeneracy may arise from other orbits yielding the same averaged weight.

\subsection{Absence of Localization in \textit{i}SWAP-class Circuits}
\label{subsec:abs_local}

This subsection examines the Pauli-orbit structure in \textit{i}SWAP-class Floquet Clifford circuits.
At fixed system size, our numerics show that initially prepared Pauli strings gain Pauli weight roughly in proportion to the number of Floquet cycles.
The trend is clean over the pre-boundary window, before the spreading front reaches the system edges.
These observations motivate a working assumption: the gap-setting eigenmode is delocalized in Pauli space, with orbit-averaged weight that is extensive in the system size.
Intuitively, start from a string with a clear left and right endpoint; as long as the front has room to grow, each cycle adds an approximately constant number of non-identity sites, so the total weight tracks the cycle count.

\subsubsection{Two-qubit Clifford gates classification}
\label{sec:classification}

\begin{figure}
    \centering
    \includegraphics[width=0.98\linewidth]{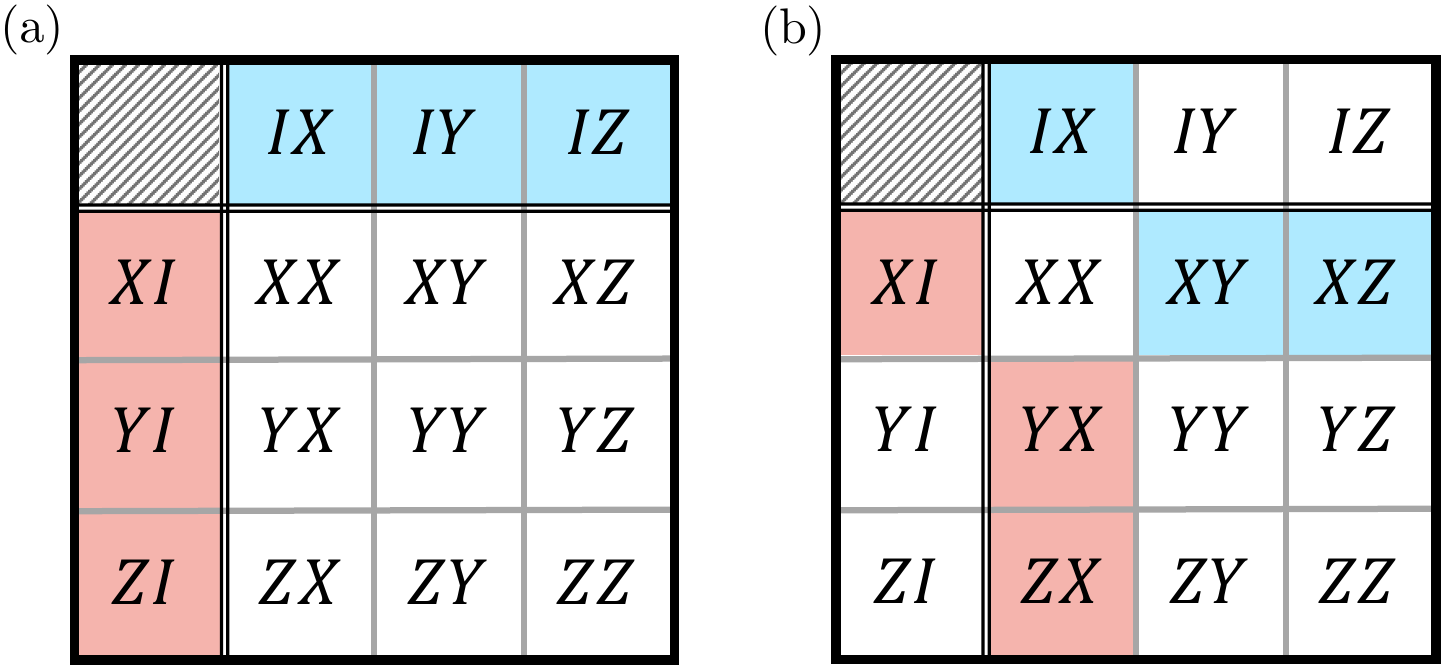}
    \caption{
    {\bf Two-qubit Clifford mapping patterns.}
    Local-Clifford classification of two-qubit Clifford gates from their action on single-site Paulis.
    The $16$ two-qubit Paulis are placed on a $4\times4$ grid, and the images of $\{X_1,Y_1,Z_1,X_2,Y_2,Z_2\}$ under conjugation are marked, with red (blue) indicating support on qubit~1 (2).
    Two patterns occur.
    (a) Trivial pattern, where single-site Paulis remain single-site supported.
    (b) Nontrivial pattern, where two single-site Paulis per qubit map to two-site Paulis within a common row or column and the third maps to the intersection single-site Pauli.
}
\label{fig:com_rel}
\end{figure}

We briefly recall the local-Clifford (LC) classification of two-qubit Clifford gates and introduce a structural invariant that distinguishes the \textit{i}SWAP class as the natural choice for avoiding orbit localization in brickwork circuits.

The LC classification can be read off from the conjugation action of a two-qubit Clifford unitary on single-site Pauli operators.
Because Clifford conjugation preserves commutation and anticommutation, any two-qubit Clifford $C$ maps the six single-site Paulis $\{X_1,Y_1,Z_1,X_2,Y_2,Z_2\}$ to six two-qubit Pauli operators with the same bipartite commutation structure.
Representing the 16 two-qubit Paulis on a $4\times4$ grid, one finds that only two distinct mapping patterns are possible.
In the trivial pattern shown in \figref{fig:com_rel}(a), the images of the single-site Paulis remain single-site supported.
In the nontrivial pattern shown in \figref{fig:com_rel}(b), for each qubit two of the three single-site Paulis are mapped to two-site Paulis that lie within a common row or column of the grid, while the remaining one is mapped to a single-site Pauli at the corresponding intersection.

Refining these two patterns by whether the map preserves or exchanges the qubit labels yields four LC equivalence classes with canonical representatives $\{I,\mathrm{SWAP},\mathrm{CZ}, i\mathrm{SWAP}\}$.
We denote the LC class of a representative $D$ by
\[
G_D := \{C \in \mathcal{C}_2 : C \simeq D\}.
\]
Within each of the two patterns in \figref{fig:com_rel}, the remaining distinction is whether the map preserves or exchanges the site labels of single-site Paulis, as indicated by the red and blue labels in the figure.
Accordingly, the pattern in \figref{fig:com_rel}(a) splits into the class $G_I$, for which the site labels are preserved, and the class $G_{\mathrm{SWAP}}$, for which the site labels are exchanged.
Likewise, the pattern in \figref{fig:com_rel}(b) splits into the class $G_{\mathrm{CZ}}$, which preserves site labels, and the class $G_{i\mathrm{SWAP}}$, which exchanges them.

In particular, this LC classification singles out two structural constraints on Pauli propagation, one ensuring the existence of strictly local modes and the other enforcing Pauli-weight conservation.
\begin{enumerate}
    \item \textit{Local support preservation}: Gates in $G_I$ and $G_{\mathrm{CZ}}$ preserve single-site support under conjugation.
    In particular, on each qubit there exists at least one single-site Pauli that remains supported on that same qubit after conjugation.
    In a brickwork circuit this can lead to commuting overlaps between consecutive layers, for instance
    \begin{align}
    \mathrm{CZ}_{12}\mathrm{CZ}_{23}
    =
    \mathrm{CZ}_{23}\mathrm{CZ}_{12},
    \nonumber
    \end{align}
    which inhibits spatial spreading and allows certain Pauli orbits to remain localized under circuit evolution, thereby generating a localized spectral sector.
    
    \item \textit{Weight preservation}: Gates in $G_I$ and $G_{\mathrm{SWAP}}$ preserve Pauli weight under conjugation. 
    That is, although a Pauli string may be transported along the chain, the number of non-identity sites is conserved and cannot increase.
    As a result, circuits built entirely from such gates cannot generate operators with extensive weight and therefore cannot support eigenmodes with extensive spatial support.
\end{enumerate}
By contrast, the \textit{i}SWAP class satisfies neither constraint. 
Under conjugation by an \textit{i}SWAP-class gate, single-site Pauli operators spread across the two qubits and can increase their Pauli weight.
As a result, brickwork circuits built from \textit{i}SWAP-class gates avoid localized sectors and allow Pauli strings to grow to extensive weight.
We therefore numerically characterize the orbit-averaged Pauli weights generated by shallow brickwork circuits built from different two-qubit Clifford gate sets.

\subsubsection{Eigenspectrum dependence on localization structure}

\begin{figure}
    \centering
    \includegraphics[width=0.98\linewidth]{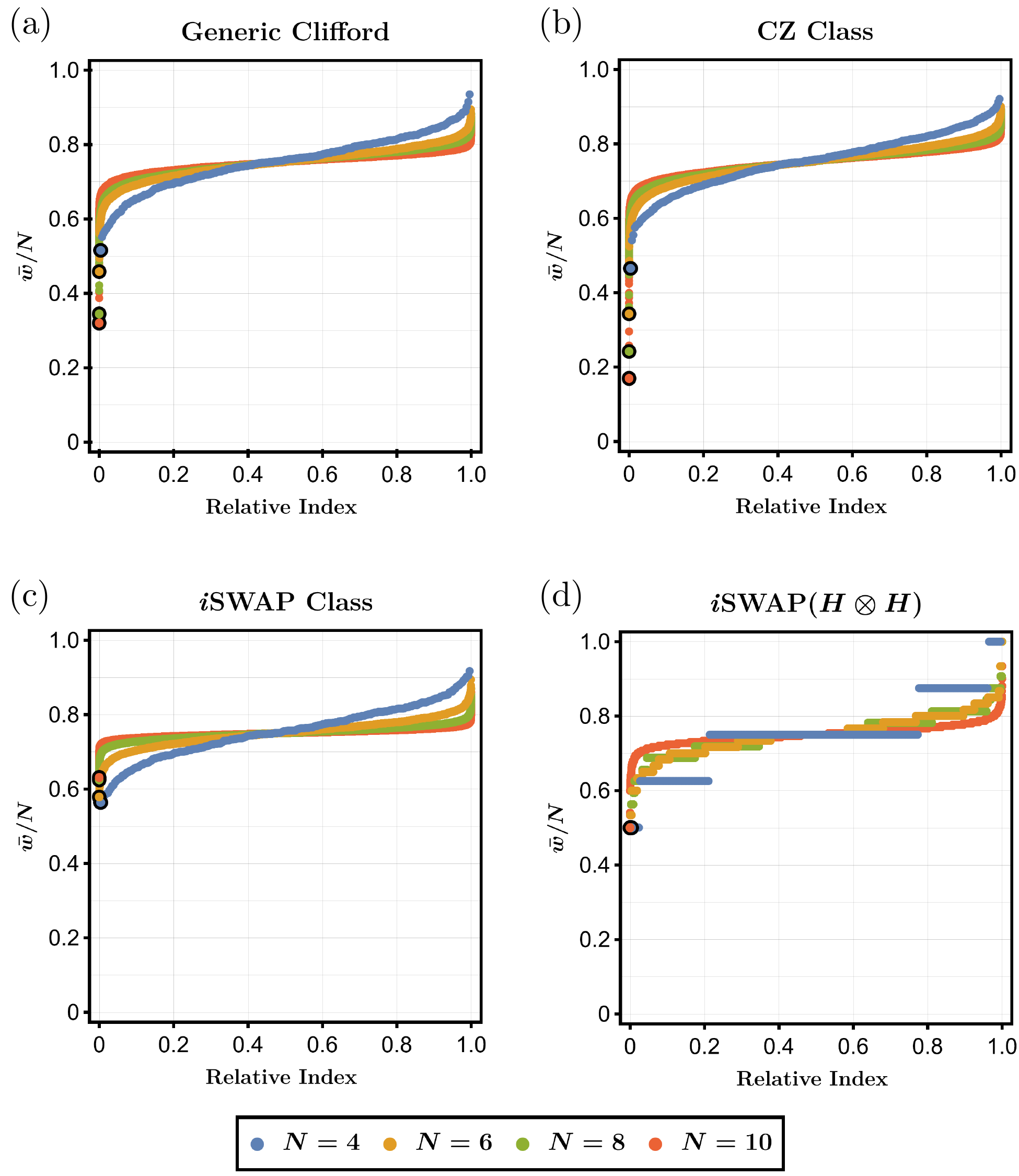}
    \caption{
    {\bf Orbit-averaged Pauli weights} for dissipative brickwork Floquet Clifford circuits.
    For each nontrivial Pauli string $S$, we compute the orbit-averaged weight $\bar w(S)$ and plot the value normalized by the system size, $\bar w(S)/N$, sorted in increasing order and labeled by the relative index.
    Black-circled markers indicate the smallest values for each $N$.
    (a) Generic two-qubit Clifford gate set.
    (b) CZ class gate set.
    (c) \textit{i}SWAP-class gate set.
    (d) \textit{i}SWAP-class circuit with all two-qubit gates fixed to $i\mathrm{SWAP}(H\otimes H)$.
    A small low-$\bar w$ sector is visible in (a) and (b), while it is absent in the \textit{i}SWAP class cases (c) and (d).}
    \label{figs:orbit}
\end{figure}

Figure~\ref{figs:orbit}(b) shows that CZ class circuits retain a small low-weight sector at the left end.
The smallest values of $\bar w(S)/N$ are highlighted by black-circled markers for each $N$, and their downward shift with increasing $N$ indicates that these strings remain localized rather than spreading to extensive weight.
Figure~\ref{figs:orbit}(a) shows the same qualitative behavior for a generic Clifford gate set, albeit with a  weaker low-weight sector.
Indeed, the localized structure is expected with probability $1/9$ in the CZ class, while in the generic Clifford ensemble used for \figref{figs:orbit}(a) it is further suppressed to $(1/9)(9/20)^2$, reflecting the probability for the relevant two-qubit gates to fall in the CZ class~\cite{kovacs2024operator}.
Through Eq.~\eqref{eq:eign_mag}, the presence of this localized weight sector $\bar w(S)=O(1)$ implies that the Liouvillian gap remains bounded by an $N$ independent constant.

We now turn to the \textit{i}SWAP class.
A low-$\bar w$ orbit can in principle arise if the single site to single site mappings align across consecutive bonds, allowing a weight-one Pauli to hop along the chain and yield an orbit with $\bar w(S)=1$.
For example, if all two-qubit gates are fixed to \textit{i}SWAP, then $Z_1\to Z_3\to\cdots\to Z_{N-1}$ and $Z_2\to Z_4\to\cdots\to Z_N$ form weight-one orbits.
However, the probability of forming such an aligned chain decays exponentially with $N$, so it does not generate a finite low-$\bar w$ sector in typical realizations.
Consistent with this, in the \textit{i}SWAP-class case in \figref{figs:orbit}(c) we find no localized sector at low $\bar w$.
In particular, the black-circled markers do not drift downward with increasing $N$, unlike in the CZ class case.
This indicates that typical Pauli orbits are not confined to a low-weight subspace but instead explore a broad range of Pauli weights.
This behavior is consistent with a Liouvillian gap that grows linearly with system size, $\Delta \sim N\gamma$.

For the remainder of this work, we fix the two-qubit gate to the representative in Eq.~\eqref{eq:iSWAP}.
This choice removes the fine-tuned weight-one chains discussed above and enforces strong spreading at the single-site level.
That is, the fixed gate is chosen so that a single-site Pauli transferred to a neighboring site (after the first Clifford layer) cannot remain single-site supported after the completion of the full Floquet step.
For example, if the first layer maps $X_1$ to $Z_2$, then the second necessarily maps $Z_2$ to a two-site Pauli operator.
As shown in \figref{figs:orbit}(d), this fixed-gate circuit exhibits the same qualitative behavior as the \textit{i}SWAP-class case in \figref{figs:orbit}(c).
In particular, the smallest normalized orbit-averaged Pauli weights take the exact value $1/2$, as depicted by the black-circled markers in \figref{figs:orbit}(d).
The corresponding orbits follow a simple alternating pattern, whose only dependence on $N$ is the string length.
Concretely, they are generated by the two-periodic strings
\begin{align}
\{\, (IY)^{\otimes N/2},\; (YI)^{\otimes N/2} \,\}.
\end{align}
As a result, for the DHFC with $n_h=0$, the Liouvillian gap takes the exact form
\begin{align}
\Delta_{\mathrm{non\text{-}doped}}=\frac{N}{2}\gamma \label{eq:gap_non-doped}
\end{align}
throughout the full dissipation range.

Therefore, in the absence of doping, the Liouvillian gap in \eqnref{eq:gap_non-doped} grows linearly in $N$ for any $\gamma$, so already the thermodynamic limit $N \rightarrow \infty$ and weak dissipation limit $\gamma \rightarrow 0^+$ do not commute.
Under the gap diagnostic, the undoped Clifford circuit exhibits the strongest intrinsic relaxation in our setting.
In some sense, this is surprising because Clifford circuits are classically efficiently simulable. The resolution is that the Liouvillian gap probes dissipative relaxation rather than unitary circuit complexity.

\section{Weak Dissipation Singularity}
\label{sec:weak_dissipation}

In this section, we study the Liouvillian gap of the DHFC under finite doping density of single-site Haar-random gates with dissipation strength $\gamma$.
We provide numerical evidence for a gap singularity across all doping densities $p_h = n_h/N$, suggesting that our model exhibits chaotic many-body dynamics.
Finally, we compare our numerics to the bounds obtained analytically for the strong dissipation regime, see \secref{sec:struc_dop}.

In \figref{figs:gap_scaling}, we show the results of our numerical simulations 
~\footnote{We use Lanczos-type algorithms to calculate the spectral gap of a given realization of the DHFC, then average over many realizations.
Standard implementations of the Lanczos algorithm are sufficient to compute $\Delta$ for $\gamma \gtrsim 1$. 
In the weak dissipation regime, we instead employ an approach similar to that outlined in \cite{yoshimura2024robustness}; starting from a uniform and traceless operator, we repeatedly apply the Floquet channel $\Phi$ and renormalize the resulting operator until sufficient convergence to a fixed operator, the slowest decaying mode, is achieved.}
for full and staggered doping configurations for $N \in \{4,6,8\}$. 
In the weak dissipation regime, \figref{figs:gap_scaling}(c) shows that the linear growth of $\langle \Delta \rangle$ with $N$ persists.
In the strong dissipation regime, \figref{figs:gap_scaling}(e) shows that the gap saturates to a finite value, contrary to the undoped case where it continues to scale linearly in $N$.
In fact, we find close agreement of our numerics with our analytic bounds for large but finite $\gamma$, particularly as $N$ increases, and we expect that the discrepancy approaches zero as $N$ increases. 
That is, when $n_h = O(N)$, the gap takes the form $\langle \Delta \rangle = A\gamma + B$ where $A,B = O(1)$. From the behavior of $\Delta$ in \figref{figs:gap_scaling}(a), it is expected that the gap increases rapidly with increasing $N$ for some small $\gamma$, saturating to a constant value.
Such behavior is unfortunately not observed for the small values of $\gamma$ accessible via our numerics, but already at intermediate values of $\gamma$, we see that the growth of $\langle \Delta \rangle$ is decreasing, which suggests the saturation of $\langle \Delta \rangle$ to the value predicted analytically.

We argue that the numerical results presented in \figref{figs:gap_scaling}, in conjunction with Result \ref{res:strong_dissipation_lower_bound}, imply a singularity in the Liouvillian gap in the thermodynamic limit for any nonzero $p_h$.
For such doping patterns, our expectation is that as the maximum system size $N$ increases, the threshold $\gamma^*$ below which the gap scales linearly in $N$ will decrease. 
Then, at any nonzero $\gamma$, the gap saturates to a finite value as $N$ increases, so $\Delta_{\gamma \rightarrow 0^+}$ will retain a finite value in the thermodynamic limit (but $\Delta_{\gamma = 0} =0$, because $\gamma=0$ corresponds to the unitary case in which the eigenvalues of the Floquet channel have modulus 1).

Thus, from the perspective of singularity diagnostics developed in recent work~\cite{mori2024liouvillian,yoshimura2025theory}, our numerics provide evidence that the dynamics of our model remain chaotic across all doping rates. 
Compared to the undoped case, at a finite doping rate, the gap retains a finite value in the thermodynamic limit, indicating a `weaker' instance of intrinsic relaxation than when $\lim_{N\rightarrow \infty} p_h = 0$.
In light of this, the following sections work in the strong dissipation regime to classify doping structures of the DHFC based on the scaling of their Liouvillian gap in the thermodynamic limit.

\section{Lower Bound for Doped Circuits via Weight Truncation}
\label{sec:low_bound}

This section bounds the Liouvillian gap from below in terms of the Haar-doping density.
We first establish in Sec.~\ref{subsec:eigen_free} a general criterion.
When the weight-truncated map supports no nontrivial eigenmodes, the Liouvillian gap of the untruncated channel is lower bounded at order \(w\gamma\).
We apply this criterion to contiguous doping structures in Sec.~\ref{subsec:cont_struc}, which yields an explicit lower bound on the gap as a function of the number of doped sites.
Finally, Sec.~\ref{subsec:gen_struc} generalizes the contiguous analysis to arbitrary doping patterns.
Using a pigeonhole argument, we derive a necessary condition on the doping density for sustaining a finite gap in the strong dissipation regime.

\subsection{Criterion for the Absence of Modes in the Truncated Subspace}
\label{subsec:eigen_free}

In Sec.~\ref{sec:non-doped}, we showed that in the undoped \textit{i}SWAP-class circuit the orbit-averaged Pauli weight is extensive in \(N\).
As a consequence, no Pauli orbit can remain confined to any fixed weight cutoff \(w_{\mathrm{t}}\).
Equivalently, for any Pauli string $S$ with \(w(S)\le w_{\mathrm{t}}\), repeated Floquet iterations must produce an iterate whose weight exceeds \(w_{\mathrm{t}}\) and eventually reaches \(O(N)\) along the orbit.
In the strong dissipation regime, we analyze the Liouvillian gap through the projected dynamics introduced in Sec.~\ref{subsec:pdynamic_trunc}, where after each step we retain only the weight below \(w_{\mathrm{t}}\) components.
Viewed in this truncated description, the finite-time escape of all low-weight operators is naturally captured by nilpotency of the projected map.
This criterion provides a convenient bridge from the undoped Clifford limit to Haar-doped circuits, where a simple orbit picture is no longer available.
Motivated by this, we introduce a notion of circuits that support no nontrivial modes in the weight-truncated subspace.
\begin{definition}[weight-\(w\) eigenmode-free]
\label{def:eigenmode-free}
Let \(\Phi\) be the Liouville-space representation of an \(N\)-qubit dissipative channel in the Pauli basis.
We say that \(\Phi\) is weight-\(w\) eigenmode-free if the projected map \(\Pi_{w}\Phi\Pi_{w}\) is nilpotent, where \(\Pi_{w}\) projects onto Pauli strings of weight at most \(w\).
\end{definition}
As an immediate consequence of Definition~\ref{def:eigenmode-free}, we can classify our undoped (fixed-gate) Clifford Floquet circuit as eigenmode-free up to an extensive weight cutoff:
\begin{corollary}
\label{cor:nondop_eigenfree}
The non-doped dissipative Clifford Floquet circuit built from the fixed two-qubit gate \(i\mathrm{SWAP}(H\otimes H)\) is weight-\((N/2-1)\) eigenmode-free.
\end{corollary}

\begin{figure}
    \centering
    \includegraphics[width=0.98\linewidth]{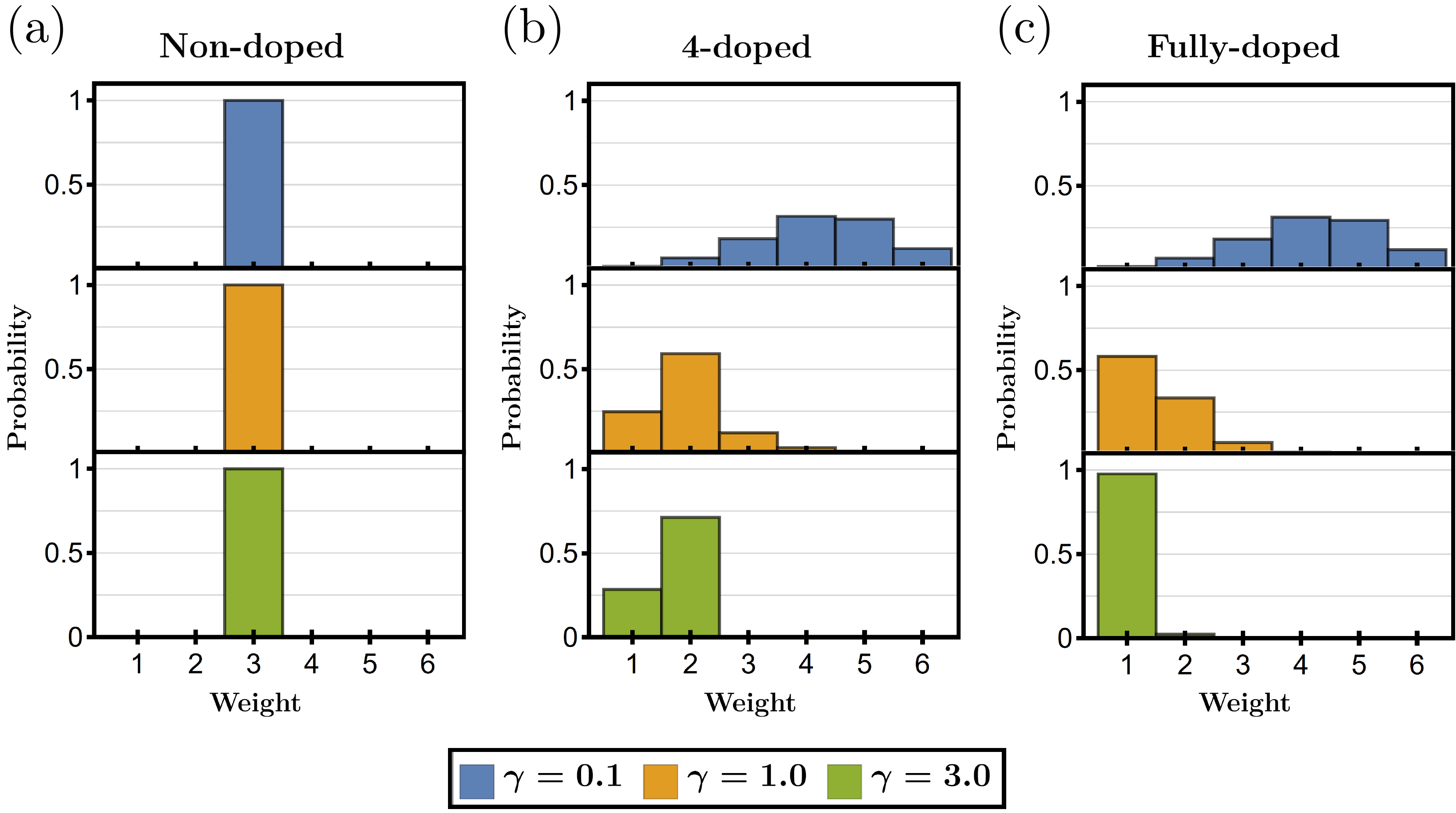}
    \caption{
    {\bf Weight distribution of the eigenoperator.}
    Pauli-weight content of the Liouvillian-gap eigenoperator for \(N=6\).
    Distributions are computed from the Pauli-basis expansion with squared-amplitude weights.
    Each panel shows \(\gamma=0.1\) (blue), \(\gamma=1.0\) (orange), and \(\gamma=3.0\) (green); \(\gamma=0.1\) corresponds to weak dissipation and \(\gamma=3.0\) to strong dissipation.
    (a) Non-doped circuit: support is concentrated at weight \(w=3\) independent of \(\gamma\).
    (b) 4-doped contiguous pattern and (c) fully doped limit; at \(\gamma=3.0\) the fully doped case concentrates at \(w=1\) whereas the 4-doped case retains a dominant \(w=2\) contribution.}

    \label{fig:weight}
\end{figure}

As already analyzed in Sec.~\ref{subsec:abs_local}, in the non-doped case \figref{fig:weight}(a) visualizes the Pauli-weight content of the Liouvillian-gap eigenoperator, obtained by averaging the Pauli-string weight with the squared magnitudes of its Pauli-basis amplitudes.
For the system size \(N=6\) shown in \figref{fig:weight}, this distribution is sharply concentrated at weight \(w=3\) (i.e., \(N/2\)) and remains unchanged as the dissipation strength is varied.
For the contiguous doping structures to be discussed in Sec.~\ref{subsec:cont_struc}, \figref{fig:weight}(b) and (c) show the corresponding weight-resolved distributions for a 4-doped pattern and for the fully doped limit, respectively, across the crossover from weak to strong dissipation.
In the weak dissipation regime, \(\gamma=0.1\) (blue), the two cases exhibit similar weight-resolved distributions.
As \(\gamma\) increases, however, the fully doped circuit collapses onto weight \(w=1\), whereas the 4-doped circuit retains substantial support across weights \(w=1\) and \(w=2\).
In particular, at \(\gamma=3.0\) (green) in the strong dissipation regime, the fully doped distribution is essentially concentrated at \(w=1\), while the 4-doped distribution retains a prominent weight-\(2\) peak, making \(w=2\) a dominant contribution absent in the fully doped limit.

Taken together, these trends suggest that the doping number controls the largest weight cutoff \(w\) for which the truncated dynamics is weight-\(w\) eigenmode-free as defined in Definition~\ref{def:eigenmode-free}.
In the non-doped circuit this eigenmode-free property extends at least up to \(w=2\), while for the 4-doped contiguous pattern it persists only up to \(w=1\).
In the fully doped limit, by contrast, a robust low-weight component survives deep in the strong dissipation regime, indicating that no finite \(w\) truncation is eigenmode-free.
This observation will be made precise for contiguous doping structures in Lemma~\ref{lem:cont_cutoff}.
More generally, once a circuit is weight-\(w\) eigenmode-free, any long-lived eigenmode cannot remain entirely within the truncated subspace and must develop support above the cutoff.
We therefore focus on how this constraint translates into a quantitative lower bound on the Liouvillian gap of the non-truncated channel.

We therefore turn to how nilpotency of the weight-truncated dynamics constrains long-lived modes in the original (non-truncated) channel \(\Phi_F\).
Weight-\(w\) eigenmode-freeness precludes eigenmodes confined entirely to weights \(\le w\), so any slowly decaying mode of \(\Phi_F\) must develop support at weights \(>w\).
With a depolarizing layer applied at the end of each period, this above-cutoff support is further suppressed by the final noise step.
Together, nilpotency in the truncated sector and the end-of-period depolarization bound the spectral radius, as stated in the following lemma.
strains the spectral radius, as stated in the following lemma.
\begin{lemma}
\label{lem:gap-bound}
Suppose a dissipative Clifford Floquet circuit \(\Phi_F\) is weight-\(w\) eigenmode-free.
If the dissipative strength \(\gamma > 0\) is sufficiently large so that \(e^{-w\gamma/2} \ll 1\) (i.e., in the strong dissipation regime),  
then the Liouvillian gap \(\Delta\) of the circuit satisfies
\begin{align}
\label{eq:lgap_weight}
\Delta \;\ge\; \left(\frac{w}{2} + 1\right)\gamma.
\end{align}
\end{lemma}

\begin{proof}
    To exploit the eigenmode-free condition, we work in a Pauli basis ordered by weight and decompose \(\tilde{\Phi}_F\) into low- and high-weight sectors.
    In this basis, the nilpotency of the projected low-weight map is encoded in the upper-left block, and we can write
    \[
        \tilde{\Phi}_{F}
        =
        \begin{bmatrix}
            \tilde{\Phi}_{w}^{\mathrm{trunc}}
                &  B
            \\
            C & D
        \end{bmatrix},
    \]
    where the upper-left block acts within Pauli strings of weight at most \(w_{\mathrm{t}}\).
    By Appendix \ref{app:Feingold_Varga}, the sperctral radius of $\tilde{\Phi}_{F}$ is upper bounded by
    \begin{align}
    \label{eq:FV_bound}
        \rho(\tilde{\Phi}_{F})
        \leq
        \rho(D)+\sqrt{\|B\|\|C\|}.
    \end{align}

    Now, recall that any block of a unitary matrix acts as a contraction.
    Since each row and column of a unitary matrix forms an orthonormal set, any submatrix must be contractive in operator norm.
    We begin by examining the block \(D\), which governs the evolution of Pauli strings with weight strictly greater than \(w\).  
    Each such string is given a suppression factor of at least \(e^{-(w+1)\gamma}\) due to the depolarizing noise, with additional decay arising from higher-weight components.  
    Factoring out this minimal suppression, we isolate a residual operator that captures the remaining dynamics within this sector.  
    If the suppression were uniform across all basis elements, this residual operator would coincide with a submatrix of the unitary channel and hence be contractive.  
    In reality, however, strings with weight greater than \(w+1\) experience even stronger damping, rendering the residual operator strictly smaller in norm.  
    As a result, the spectral radius of \(D\) is bounded as
    \[
        \rho(D) \leq e^{-(w+1)\gamma}.
    \]

    We now turn to the off-diagonal blocks \(B\) and \(C\), which mediate transitions between subspaces of weight less than \(w\) and those of higher weight.
    Unlike \(D\), these blocks are not associated with square subspaces, and the amount of noise experienced differs between the input and output strings.  
    For \(B\), the input Pauli strings have weight \(w\), while the output strings have weight at least \(w+1\), resulting in an overall suppression factor of at least \(e^{-(\frac{w}{2}+1)\gamma}\).  
    The same reasoning applies to \(C\), with the roles of input and output reversed.  
    This yields the operator norm bounds
    \[
        \|B\| \leq e^{-(\frac{w}{2}+1)\gamma}, \qquad \|C\| \leq e^{-(\frac{w}{2}+1)\gamma}.
    \]

    Combining the bounds on \(D\), \(B\), and \(C\) and the Feingold-Varga inequality in Eq.~\eqref{eq:FV_bound}, we obtain
    \[
        \rho(\tilde{\Phi}_F) 
        \leq e^{-(\frac{w}{2}+1)\gamma} \left(1 + e^{-\frac{w}{2}\gamma} \right).
    \]
    In the strong dissipation regime, where \(e^{-w\gamma/2} \ll 1\), this simplifies to
    \[
        \rho(\tilde{\Phi}_F) 
        \leq e^{-(\frac{w}{2}+1)\gamma}.
    \]
    Recalling that the Liouvillian gap is defined as \(\Delta := -\log \rho(\tilde{\Phi}_F)\),  
    we obtain the lower bound as claimed in Eq.~\eqref{eq:lgap_weight}.
\end{proof}

Thus, in the strong dissipation regime, weight-\(w_{\mathrm{t}}\) eigenmode-freeness forces the Liouvillian gap to be at least of order \(w_{\mathrm{t}}\gamma\).
As shown in Corollary~\ref{cor:nondop_eigenfree}, the undoped DHFC is eigenmode-free up to an \(N\)-order cutoff.
Lemma~\ref{lem:gap-bound} then yields an \(N\)-order lower bound on the gap \(\Delta\) in the strong dissipation regime, consistent with the exact result \(\Delta=(N/2)\gamma\) established in Sec.~\ref{subsec:abs_local}.
Haar doping mixes Pauli strings that would otherwise belong to distinct orbits of the non-doped circuit, thereby breaking the Pauli-orbit structure.
As a result, a single Pauli string typically splits into a linear combination of Pauli strings of comparable weight, so nilpotency of the weight-\(O(N)\) truncated dynamics is no longer guaranteed.
In the remainder of this section, we quantify how the doping rate determines eigenmode-freeness and hence the gap bound.

\subsection{Contiguous Doping Structure}
\label{subsec:cont_struc}

We begin with a simple structured pattern in which the Haar-doped qubits form a single connected block along the ring.
We call this a contiguous doping structure since it partitions the chain into a doped block and an undoped block.
These two regions play distinct roles in the truncated dynamics.
The undoped segment drives weight growth and enlarges the eigenmode-free cutoff, whereas the doped segment mitigate this through Pauli mixing.
This competition is captured by the following lemma, which shows how the eigenmode-free cutoff weight scales with the number \(n_{\mathrm{cl}}\) of undoped sites.
\begin{lemma}
\label{lem:cont_cutoff}
Consider a dissipative Haar-doped Clifford circuit on \(N\) qubits with contiguous doping.
Let \(n_{\mathrm{cl}}\) be the number of undoped qubits.
Then the one-period channel is weight-\(\lfloor c\,n_{\mathrm{cl}}\rfloor\) eigenmode-free for some constant \(c>0\) independent of \(N\).
\end{lemma}
\begin{proof}
    We begin by expressing the $k$-step channel restricted to low-weight Pauli operators.
    As shown in Appendix~\ref{app:clifford-decomp}, any single-qubit Haar-doped gate can be written as a real linear combination of Clifford conjugation maps.
    This yields a representation of the $k$-step channel as a linear combination over \textit{i}SWAP-class Clifford circuit sequences of length $k$, each determined by the realization of Haar-doped gates at each step.
    Letting $\Pi_{w}$ denote the projector onto Pauli strings of weight at most $w$, we write the truncated evolution as
\begin{align}
(\tilde{\Phi}_{w_{\mathrm{t}}}^{\mathrm{trunc}})^k 
=
\sum_{\vec{i}} a_{\vec{i}}\, \Pi_{w} \mathcal{C}_{i_1} \Pi_{w} \mathcal{C}_{i_2} \cdots \Pi_{w} \mathcal{C}_{i_k} \Pi_{w},
\end{align}
    where each $\mathcal{C}_{i_j}$ is a Clifford channel in the \textit{i}SWAP class, and the coefficients $a_{\vec{i}} \in \mathbb{R}$ depend on the Haar-doped gates.
    Here each \(\mathcal{C}_{i_j}\) is an \emph{i}SWAP-class Clifford channel obtained by expanding the Haar-doped single-qubit rotations at step \(j\) into a real linear combination of single-qubit Clifford conjugation maps as in Appendix~\ref{app:clifford-decomp}.
    The coefficient \(a_{\vec{i}}\) is the product of the corresponding expansion coefficients along the sequence \(\vec{i}=(i_1,\dots,i_k)\).
    Equivalently, \(\mathcal{C}_{i_j}\) is the brickwork built from the fixed two-qubit gate in Eq.~\eqref{eq:iSWAP}, dressed by single-qubit Clifford gates on the doped sites, with the dressing specified by \(i_j\).
    On doped lines these single-qubit Clifford dressings can vary across steps and across terms in the expansion, whereas on undoped lines no dressing is present and the two-qubit gate acts identically in every \(\mathcal{C}_{i_j}\).

    We now turn to Pauli-weight growth under the projected evolution.
    Along any term in the expansion, the support still spreads ballistically in the \emph{i}SWAP class, but the weight growth depends on whether the string propagates through the undoped or doped region.
    In the doped region, the Clifford dressing can vary from step to step, so weight growth is not guaranteed at every iteration and some components can stay within the cutoff.
    By contrast, in the undoped region the dynamics is fixed from step to step, and a Pauli string that enters this region is driven to higher weight at a nonzero rate.
    Consequently, after sufficiently many steps \(k\), any initially low-weight operator accumulates weight of at least \(c\,n_{\mathrm{cl}}\) for some constant \(0<c<1\).
    
    It follows that if the truncation threshold \(w\) is chosen strictly below \(c\,n_{\mathrm{cl}}\), then any Pauli string with initial weight at most \(w\) is driven above the cutoff after \(k\) projected steps.
    Equivalently, the projected evolution annihilates the entire weight-\(\le w\) subspace after finitely many iterations, so that$(\tilde{\Phi}^{\mathrm{trunc}}_w)^k = 0$.
    This is precisely the nilpotency condition in Definition~\ref{def:eigenmode-free}.
    Taking \(w=\lfloor c\,n_{\mathrm{cl}}\rfloor\) shows that the circuit is weight-\(\lfloor c\,n_{\mathrm{cl}}\rfloor\) eigenmode-free.
\end{proof}

Lemma~\ref{lem:cont_cutoff} shows that in the contiguous case the eigenmode-free cutoff weight scales proportionally with the size \(n_{\mathrm{cl}}\) of the undoped region.
Inserting this cutoff into the general gap bound of Lemma~\ref{lem:gap-bound} yields the following \(n_{\mathrm{cl}}\)-dependent lower bound on the Liouvillian gap.
\begin{corollary}
\label{col:cont_gap-ncl}
For the dissipative Haar-doped Clifford circuit with a contiguous doping pattern, let \(n_{\mathrm{cl}}\) denote the number of undoped lines.
In the strong dissipation regime, the Liouvillian gap \(\Delta\) obeys
\begin{align}
\Delta \geq \left(\frac{\lfloor c\,n_{\mathrm{cl}} \rfloor}{2} + 1\right)\gamma,
\end{align}
for some constant \(0<c<1\) independent of \(N\).
\end{corollary}
Corollary~\ref{col:cont_gap-ncl} bounds the gap from below by a term proportional to \(n_{\mathrm{cl}}\gamma\) in the contiguous setting.
This indicates that suppressing the non-doped, \(N\)-order singular response in the thermodynamic limit requires sufficiently dense Haar doping.
We now make this quantitative by converting an \(N\)-independent bound on \(\Delta\) into a necessary condition on the doping rate.
\begin{lemma}
\label{lem:cont_lgap_dop}
Consider the dissipative Haar-doped Clifford circuit with contiguous doping in the strong dissipation regime.
Assume that the Liouvillian gap is bounded uniformly in \(N\) by
\[
\Delta \le A\,\gamma + B,
\]
with constants \(A,B\) independent of \(N\).
Then the number of Haar-doped qubits must satisfy
\[
n_h \ge N - n_0,
\]
where \(n_0>0\) is an \(N\)-independent constant.
\end{lemma}
\begin{proof}
Combining Corollary~\ref{col:cont_gap-ncl} with the assumed \(N\)-independent upper bound \(\Delta \le A\gamma + B\), we obtain
\[
\left(\frac{\lfloor c\,n_{\mathrm{cl}} \rfloor}{2} + 1\right)\gamma < A\gamma + B.
\]
In the strong dissipation regime, \(\gamma\) is taken large enough that \(B/\gamma\) is negligible compared to constants, so this inequality forces
\(
\lfloor c\,n_{\mathrm{cl}}\rfloor \lesssim 2A
\),
and hence \(n_{\mathrm{cl}}\le n_0\) for an \(N\)-independent constant $n_0 := 2A/c$, where \(c>0\) is the constant from Corollary~\ref{col:cont_gap-ncl}.
Since \(n_{\mathrm{cl}}:=N-n_h\), this implies \(n_h \ge N-n_0\), as claimed.
\end{proof}

In the thermodynamic limit \(N\to\infty\), this implies that the Haar-doping rate satisfies \(p_h\approx 1\) for contiguous doping.
In other words, sustaining a system-size independent gap requires the undoped segment to remain of constant length in the strong dissipation regime.
We next remove the contiguity assumption and ask how this constraint generalizes to arbitrary doping patterns.

\subsection{Necessary Doping Rate for a Finite Gap}
\label{subsec:gen_struc}

We now generalize the doping pattern by allowing the \(n_h\) Haar-doped qubits to be placed at arbitrary sites on the periodic chain.
In this setting, the undoped sites form several disconnected arcs around the ring.
We denote by \(L_{\max}\) the length of the longest undoped arc, namely the maximum number of consecutive undoped qubits between successive doped sites under periodic identification.
A simple pigeonhole argument shows that \(L_{\max}\) must be large for a given \(n_h\), which we formalize in the following lemma.
\begin{lemma}
\label{lem:largest-undoped}
Consider a periodic chain of \(N\) qubits with \(n_h\ge 1\) Haar-doped sites placed at arbitrary positions.
Let \(L_{\max}\) be the length of the longest contiguous undoped arc on the ring.
Then
\begin{equation}
\label{eq:bound-n-arb}
L_{\max} \;\ge\; \left\lceil \frac{N-n_h}{n_h}\right\rceil.
\end{equation}
For completeness, if \(n_h=0\) then \(L_{\max}=N\).
\end{lemma}
\begin{proof}
The $n_h$ doped sites partition the $N-n_h$ undoped sites into exactly $n_h$ contiguous arcs around the ring. Let these arc lengths be $L_1,\dots,L_{n_h}$, so that $\sum_{j=1}^{n_h} L_j = N - n_h$. By the pigeonhole principle,
\(
\max_j L_j \ge (N-n_h)/n_h,
\)
and taking the ceiling yields the claim. Near-uniform placement makes all $L_j$ differ by at most $1$, which saturates the bound up to $\pm 1$.
For $n_h=0$ the statement is trivial.
\end{proof}

The longest undoped arc \(L_{\max}\) plays the same role in the general setting as the undoped segment length in the contiguous case.
Since this arc provides the largest region over which the fixed-gate Clifford dynamics can drive weight growth, it determines the gap lower bound.
This leads to the following bound.
\begin{theorem}[Liouvillian gap lower bound]
\label{thm:gen_gap-ncl}
Consider the dissipative Haar-doped Clifford circuit  with \(n_h\geq1\) doped sites placed arbitrarily.
In the strong dissipation regime, the Liouvillian gap \(\Delta\) obeys
\begin{align}
\label{eq:gen_lbound}
\Delta > c\left(\frac{N}{n_{\mathrm{h}}}-1\right)\gamma,
\end{align}
for some constant \(0<c<1\) independent of \(N\).
\end{theorem}
\begin{proof}
For an arbitrary doping pattern, the longest undoped arc \(L_{\max}\) plays the same role as the undoped segment length in the contiguous case, since it provides the largest region governed by the fixed-gate Clifford dynamics.
Applying Lemma~\ref{lem:cont_cutoff} to this arc shows that the one-period channel is weight-\(\lfloor c\,L_{\max}\rfloor\) eigenmode-free for some constant \(c>0\).
Applying Lemma~\ref{lem:gap-bound} with \(w=\lfloor c\,L_{\max}\rfloor\) gives
\[
\Delta \;\ge\; \left(\frac{\lfloor c\,L_{\max}\rfloor}{2}+1\right)\gamma.
\]
Lemma~\ref{lem:largest-undoped} lower bounds the longest undoped arc as \(L_{\max}\ge \lceil (N-n_h)/n_h\rceil\).
Substituting this estimate and simplifying yields Eq.~\eqref{eq:gen_lbound}, where \(c\) is half of the constant from Lemma~\ref{lem:cont_cutoff}.
\end{proof}

Theorem~\ref{thm:gen_gap-ncl} expresses the gap lower bound directly in terms of the number of doped sites \(n_h\).
We can therefore read off immediate consequences for how $n_h$ must scale with $N$ in order to guarantee that the gap remains bounded in the thermodynamic limit.
\begin{corollary}
\label{cor:sub_dop_div_gap}
In the strong dissipation regime, if the number of Haar-doped qubits $n_h$ grows sublinearly with the system size $N$, then the Liouvillian gap diverges in the thermodynamic limit.
\end{corollary}
This follows immediately from Theorem~\ref{thm:gen_gap-ncl}, since sublinear growth of $n_h$ implies $N/n_h\to\infty$ and thus Eq.~\eqref{eq:gen_lbound} rules out an $N$-independent bound on the Liouvillian gap.
Rephrasing Corollary~\ref{cor:sub_dop_div_gap} in the complementary direction gives:
\begin{theorem}[Necessary doping rate for finite gap]
\label{thm:fin_gap_requires_lin_dop}
In the strong dissipation regime, if the Liouvillian gap remains bounded as $N\to\infty$, then the number of Haar-doped sites must scale linearly with \(N\).
\end{theorem}

Corollary~\ref{thm:fin_gap_requires_lin_dop} is one of our main results, as it gives a necessary scaling of the Haar-doping for the Liouvillian gap to remain bounded as \(N\to\infty\).
A representative instance of such bounded-gap behavior is the fully doped circuit, which in the strong dissipation regime yields \(\Delta=\gamma+2\); we analyze this limit in Sec.~\ref{subsec:ful_dop}.
The corollary thus constrains the crossover from the non-doped regime, characterized by an \(N\)-order singular response, to a bounded-gap regime of the fully doped type.
In particular, bounded relaxation already requires a linear number of doped sites, equivalently a nonvanishing doping rate in the thermodynamic limit.
Still, linear Haar doping is only a necessary condition and does not by itself guarantee a bounded gap.
In the next section, we study structured patterns with \(n_h\propto N\) and show that they indeed realize a system size independent gap in the strong dissipation regime.

\section{Structured Doping for a Finite Gap}
\label{sec:struc_dop}

We now focus on doping structures with a nonvanishing rate, $n_{\mathrm{h}}\propto N$.
Within this setting, we show that there exists a range of structured doping arrangements for which the gap remains independent of the system size in the strong dissipation regime.
We proceed by analyzing three constructions that capture this behavior.
In \secref{subsec:ful_dop}, we start from the fully doped limit, where the Liouvillian gap can be obtained exactly.
Section~\ref{subsec:dense_dop} turns to dense doping, where every undoped segment is separated by at least one doped site. We solve the minimally dense case exactly and derive an upper bound for more general arrangements.
In \secref{subsec:block_dop}, we study block-staggered structures with regularly spaced undoped blocks between doped sites.

\subsection{Fully-doped Circuit}
\label{subsec:ful_dop}

We therefore apply the effective Pauli dynamics via weight truncation introduced in Sec.~\secref{subsec:pdynamic_trunc}, and begin with the minimal nontrivial cutoff \(w_{\max}=1\).
This restricts the dynamics to the subspace \(\mathcal{P}_{\le 1}\) spanned by single-site Pauli operators, a \(3N\)-dimensional space once the identity (the steady state) is excluded.
In this setting, the truncated propagation operator \(\tilde{\Phi}_1^{\mathrm{trunc}}\) reduces to a generalized permutation matrix, rendered analytically tractable by the uniform mixing induced by Haar conjugation and the diagonal damping from local depolarization.
This structure enables an exact determination of the Liouvillian gap in the thermodynamic limit, without any Haar averaging, as stated below.

\begin{theorem}
\label{thm:ful_lgap}
Consider the dissipative Haar-doped Floquet Clifford circuit in the fully doped structure \(n_h=N\).
In the strong dissipation regime where the \(w_{\max}=1\) truncation is valid, the Liouvillian gap is
\begin{align}
\label{eq:lgap_ful}
\Delta_{\mathrm{full}}=\gamma+2
\end{align}
in the thermodynamic limit \(N\to\infty\).
\end{theorem}

\begin{proof}
We consider the truncated propagation operator \(\tilde{\Phi}_1^{\mathrm{trunc}}\), which acts on the weight-one Pauli subspace.
Under the fixed gate defined in Eq.~\eqref{eq:iSWAP}, only Pauli strings of the form \(X_i\) remain within the subspace after one Clifford layer.
All other strings are mapped to weight-two operators and thus removed by truncation.

Assuming an even system size \(N\) with periodic boundary conditions, the two-layer brickwork structure causes a clean separation between even and odd sublattices.
Specifically, an \(X_i\) operator at an even (odd) site is conjugated to \(X_{i+2}\) (\(X_{i-2}\)).
This structure implies that \(\tilde{\Phi}_1^{\mathrm{trunc}}\) takes the form
\[
\tilde{\Phi}_{1}^{\mathrm{trunc}}
= \tilde{\Phi}_{1,+}^{\mathrm{trunc}} \oplus \tilde{\Phi}_{1,-}^{\mathrm{trunc}} .
\]
Here \(\tilde{\Phi}_{1,\pm}^{\mathrm{trunc}}\) act on weight-one strings supported on the even and odd sublattices, respectively.
Introduce the sublattice bases \(\{\ket{j}_\pm\}_{j=1}^{N/2}\) by \(\ket{j}_+ \equiv \ket{2j}\) and \(\ket{j}_- \equiv \ket{2j-1}\), with \(j\) understood modulo \(N/2\).
In this notation,
\[
\tilde{\Phi}_{1,\pm}^{\mathrm{trunc}}
:= e^{-\gamma}\sum_{j=1}^{N/2} \ket{j\pm 1}_{\pm}\bra{j}_{\pm} \otimes M^{\pm}_j,
\]
where \(M_j^{\pm}\) act on the single-qubit Pauli space \(\{X,Y,Z\}\).

Let \(V_{1;j}\) and \(V_{2;j}\) denote the single-qubit Haar-random unitaries applied at site \(j\) in the first and second Haar layers, respectively.
We also introduce the corresponding Pauli transfer-matrix element
\[
    v_{k;j}:=\frac{1}{2}\Tr\!\left[V_{k;j} Z V_{k;j}^\dagger X\right], \qquad k\in\{1,2\},
\]
which is a real random variable under Haar measure.
Each \(M_j^\pm\) is a \(3\times 3\) matrix with support only in its first column. In particular,
\[
(M_j^+)_{11}:=v_{1;2j+1}v_{2;2j+2},
\quad
(M_j^-)_{11}:=v_{1;2j-2}v_{2;2j-3}.
\]
Only the \((1,1)\) entries affect the spectral radius.

Raising \(\tilde{\Phi}_1^{\mathrm{trunc}}\) to the power \(N/2\) closes the translation cycles on each sublattice. As a result, \((\tilde{\Phi}_1^{\mathrm{trunc}})^{N/2}\) is block-diagonal, with one block for the even sector and one for the odd sector, each given by an ordered product of the corresponding \(M_j^{\pm}\).
Since each block has rank one, its spectral radius is set by the magnitude of its \((1,1)\) entry.
Taking the \(N/2\)-th root then gives the spectral radius of \(\tilde{\Phi}_1^{\mathrm{trunc}}\), and hence the Liouvillian gap,
\begin{equation}
\label{eq:lgap_fulldoped}
\Delta
=\gamma-\frac{2}{N}\max_{\eta\in\{\mathrm{even},\mathrm{odd}\}}
\sum_{j\in\eta}\log\!\left|v_{1;j}\,v_{2;j+1}\right|.
\end{equation}

In the large-\(N\) limit, the maximization over the even and odd subsectors becomes immaterial, since both subsectors sample the same i.i.d.\ statistics of the Haar variables.
Indeed, in the fully doped structure the quantities \(\log|v_{1;j}v_{2;j+1}|\) appearing in Eq.~\eqref{eq:lgap_fulldoped} are independent and identically distributed across \(j\), with \(v_{1;j}\) and \(v_{2;j}\) drawn independently from the same Haar-induced distribution.
By the law of large numbers,
\begin{align}
\label{eq:haar_avg_ful}
\lim_{N\to\infty}
\frac{2}{N}\sum_{j\in\eta}\log|v_{1;j}v_{2;j+1}|
=
\mathbb{E}\!\left[\log|v_{1}v_{2}|\right]
=2\,\mathbb{E}\!\left[\log|v|\right],
\end{align}
for \(\eta\in\{\mathrm{even},\mathrm{odd}\}\), where \(v\) denotes a single draw of \(\frac12\Tr[V Z V^\dagger X]\) with \(V\sim\mathrm{Haar}\).
Using Appendix~\ref{app:haar-average}, we have \(\mathbb{E}[\log|v|]=-1\).
Substituting this into Eq.~\eqref{eq:lgap_fulldoped} yields Eq.~\eqref{eq:lgap_ful}.
\end{proof}

Note that the thermodynamic-limit gap in Eq.~\eqref{eq:lgap_ful} does not rely on averaging the Floquet channel over Haar-random single-qubit rotations on the doped sites.
Rather, the thermodynamit limit result follows from self-averaging within a typical circuit realization, where the product of independent Haar-induced factors concentrates, as derived in Eq.~\eqref{eq:haar_avg_ful}.
For context, a closely related solvable setting is the dissipative random phase model of Ref.~\cite{yoshimura2024robustness}, which provides an exactly solvable example of chaotic Floquet dynamics with a weak dissipation singular response. 
In the large on-site dimension limit the Liouvillian gap can be obtained analytically as \(\Delta=\varepsilon(2+\gamma)\); for unit random-phase strength \(\varepsilon=1\) this takes the same strong dissipation form as our fully doped qubit result.

In our truncation-based picture, it is natural to view the fully doped structure as providing the smallest Liouvillian gap among all DHFC doping patterns. Full Haar doping most effectively sustains the low-weight sector, while inserting undoped sites introduces additional leakage out of this sector and can only accelerate relaxation.
In this sense, \(\Delta_{\mathrm{full}}\) serves as a lower-bound baseline for the structured dopings analyzed below.

\subsection{Dense-Doped Structure}
\label{subsec:dense_dop}

We focus on the dense-doping regime, defined by the spacing constraint that no two undoped qubits are adjacent.  
A particularly tractable case within this regime is the staggered doping configuration, where every even site is doped (\(n_h = N/2\)) and every odd site is undoped (or vice versa). In this structure, the Liouvillian gap takes the form
\begin{align}
\label{eq:stag_lgap_0}
\Delta_{\mathrm{stag}} = 2\gamma + 3
\end{align}
in the thermodynamic limit \(N \to \infty\), provided that
the dissipation rate $\gamma$ is strong enough to truncate the
dynamics to the weight-two Pauli sector.
The explicit derivation is provided in Appendix~\ref{app:struc_staggered}, as the weight-2 truncated subspace is closed under the Floquet evolution, forming a cyclic structure among operators of the form \(Y_{2i}X_{2i+2}\). Each of the resulting \(N/2\)-site cycles involves three independent Haar unitaries act at each Floquet step, which gives rise to the constant term in Eq.~\eqref{eq:stag_lgap_0}.
\\

We now turn to the broader staggered-like doping configuration, where undoped qubits are separated by one or two doped sites. In this structure, we can analytically characterize the Liouvillian gap within a bounded range, as stated in the following lemma.
\begin{lemma}[Staggered-like doping]
Consider a dissipative Floquet Clifford circuit with a staggered-like doping configuration, 
where every undoped qubit is separated by one or two doped qubits. 
Then the Liouvillian gap $\Delta_{\mathrm{stag}^*}$ takes the form
\begin{align}
\label{eq:up_stag_stag}
\Delta_{\mathrm{stag}^*} \leq 3\gamma + 3,
\end{align}
provided that the dissipation rate $\gamma$ is strong enough to truncate the dynamics to the weight-3 Pauli sector.
\end{lemma}
\begin{proof}
    We now show that the truncated Floquet evolution \(\tilde{\Phi}^{\mathrm{trunc}}_F\) is non-nilpotent and obtain a quantitative lower bound on its spectral radius.  
    Our strategy is to analyze the return amplitude \(\mathrm{Tr}\!\left[P\,(\tilde{\Phi}^{\mathrm{trunc}}_F)^L(P)\right]\) for a normalized Pauli string \(P\), and to relate its long-time growth to the spectral radius via Gelfand’s formula.
    When this amplitude exhibits coherent accumulation across iterations—enabled by the existence of a non-nilpotent cycle in the support graph—its exponential growth rate yields a lower bound on the spectral radius.
    A detailed derivation is provided in Appendix~\ref{app:Lgap_upp_bound}; below, we outline the key steps.

    The first step is to characterize which Pauli strings can participate in nonzero return cycles.
    Appendix~\ref{app:struc_staggered_like_cycle} establishes the existence of such cycles by explicitly constructing a single weight-3 return cycle of period \(L = N/2\) in the truncated dynamics.
    In this cycle, every Pauli string has support on three consecutive sites, in the sense that all nontrivial Pauli operators are confined to a contiguous block of length three.
    Among the elements of this cycle, we fix a reference Pauli string \(P\) of the form
    \(P = Y_i X_{i+1} Z_{i+2}\) for some site \(i\) at which such a configuration occurs.
    
    Using this reference operator \(P\), we classify all other return cycles of period \(L\) that pass through it.
    Under repeated applications of the Floquet map, any Pauli string whose support extends beyond three consecutive sites eventually enters an undoped region, at which point its weight exceeds three and the corresponding component is truncated.  
    Consequently, only Pauli strings confined to three consecutive sites can survive as nonzero elements of a return cycle.  
    Table~\ref{tab:local_trans} and \figref{fig:local_trans} enumerate all allowed transitions among such strings, revealing multiple possible return paths.  
    This classification allows us to systematically identify all return cycles of length \(L = N/2\) involving \(P = YXZ\), each contributing to the non-nilpotent dynamics within the truncated subspace.

    Finally, to complete the bound, we compute the return amplitude 
        \(
        \mathrm{Tr}\!\left[
        P\, \bigl( \tilde{\Phi}^{\mathrm{trunc}}_F \bigr)^L (P)
        \right]
        \)
    for \( P = YXZ \), following the dominant return paths classified above.  
    As detailed in Appendix~\ref{app:Lgap_upp_bound}, the amplitude factorizes into a product of contributions accumulated along each step of the return cycle.  
    Some of these terms depend on the specific return path, while others arise uniformly at each Floquet step from single-qubit Haar-random unitaries acting at fixed locations.

    Applying Gelfand’s formula to the correlation function evaluated along the return cycle, we obtain a lower bound on the spectral radius in terms of the accumulated transition amplitudes.  
    The corresponding Liouvillian gap is therefore upper bounded as
        \[
        \Delta_{\text{stag}^*}
        \le 3\gamma + c,
        \]
    where the constant \(c\) admits an approximate expression in the limit of large \(L\), provided that both \(\alpha(L)\) and \(\beta(L)\) scale linearly with \(L\):
        \[
        c = \bar{\alpha}(4 - 2\log 2) + \bar{\beta},
        \]
    with \(\bar{\alpha} := \alpha(L)/L\) and \(\bar{\beta} := \beta(L)/L\), denoting the densities of path-dependent and path-independent contributions, respectively. 
    
    In the thermodynamic limit, among staggered-like return motifs with dissipation slope larger than \(2\gamma\), the largest averaged constant is achieved by the repeating pattern (ii$'$)$\to$(iii)$\to$(i$'$) in \figref{fig:local_trans}, giving
    \(c_{\max}=(10-2\ln 2)/{3}\approx 2.9<3\).
Therefore \(c\le 3\), which yields Eq.~\eqref{eq:up_stag_stag}.
\end{proof}

    These results establish that, under strong dissipation, the Liouvillian gap in staggered-like doping configurations remains finite and independent of system size.  
    This size-independence arises from the strict locality of return cycles, which confines non-nilpotent dynamics to three-site regions.  
    The same mechanism applies more broadly to denser doping patterns, as formalized in the following lemma.
\begin{theorem}[Dense doping]
\label{thm:dense}
Consider a dissipative Haar-doped Floquet Clifford circuit in which every undoped qubit is separated by at least one doped qubit.  
In the limit of strong dissipation \(\gamma \to \infty\), the Liouvillian gap \(\Delta_{\mathrm{dense}}\) is upper bounded as
\begin{align}
    \Delta_{\mathrm{dense}} \leq 3\gamma + 3.
\end{align}
\end{theorem}
\begin{proof}
    Pauli string transformations induced by doped unitaries always contain the same Pauli components that would appear under the corresponding undoped transformations, albeit weighted by additional Haar-random coefficients whenever the local operator is nontrivially rotated.  
    This structural overlap allows segments containing two or more consecutive doped qubits to be partially reinterpreted as effectively undoped, enabling a decomposition of the full circuit into a sum over staggered-like configurations.  
    Since each such configuration admits return paths of length \(L = N/2\) through the reference operator \(P\), the total return amplitude \(\mathrm{Tr}\!\left[P\,(\tilde{\Phi}^{\mathrm{trunc}}_F)^L(P)\right]\) can be written as a linear combination of contributions from these substructures.

    Among the cycles generated by staggered-like doping configurations, most consist of weight-3 Pauli strings (see \figref{fig:local_trans}).
    However, a particular substructure of the form \(\big[\!\bullet\!\circ\!\;\big]\) allows the local transformation \(YIX \mapsto YIX\) or \(YIZ\), thereby introducing weight-2 Pauli strings into the cycle.  
    Because of their lower weight, these components acquire only a dissipation factor of \(e^{-2\gamma}\) at each Floquet step, in contrast to the typical \(e^{-3\gamma}\) factor associated with weight-3 strings.  
    In the strong dissipation limit \(\gamma \to \infty\), any cycle containing such a weight-2 segment therefore dominates the total return amplitude.

    Among all staggered-like configurations, the dominant contribution is attained by the following patterns.  
    If there are \(2k\) undoped qubits in total, then the configuration \(\big[(\bullet\circ)^{k-1}\bullet\bullet\big]\) yields the leading contribution to the return cycle.  
    Similarly, if the number of undoped qubits is odd, i.e., \(2k+1\), the dominant structure becomes \(\big[\bullet(\bullet\circ)^{k-1}\bullet\bullet\big]\).
    Note that any other staggered-like configuration leads to a strictly larger average Pauli weight along the cycle, and hence a smaller contribution to the return amplitude.
    
    This shows that, in the strong dissipation limit, the spectral radius is controlled by these dominant cycles whose effective weight never exceeds 3, and therefore the Liouvillian gap is bounded from above by a term of the form \(3\gamma + 3\), completing the proof.
\end{proof}

Taken together, these results show that under the dense-doping condition, where every undoped site is separated by at least one doped site, the Liouvillian gap admits a system-size-independent upper bound in the strong dissipation regime.
Even as the number of undoped sites grows with \(N\), the intervening doped sites confine the dominant return cycles to finite effective weight and prevent any deterioration of the relaxation scale in the thermodynamic limit.

\subsection{Block-staggered Structure}
\label{subsec:block_dop}

In Sec.~\ref{subsec:dense_dop} we focused on dense doping, where undoped qubits do not form extended regions and the Liouvillian gap remains independent of the system size \(N\) in the strong dissipation regime.
Here we turn to sparser regular dopings in which undoped sites appear in contiguous blocks.
We consider the block-staggered structure, defined by a block length \(k\ge 1\) and the repeating pattern
\[
\big[(\circ)^{k}\bullet(\circ)^{k}\bullet\cdots(\circ)^{k}\bullet\big],
\]
so that each doped site is separated from the next by a block of \(k\) undoped sites.
The corresponding doping density is \(p_h=1/(k+1)\), and our aim is to obtain a system-size-independent upper bound on the gap in the thermodynamic limit.

Since block-staggered patterns contain even more undoped sites than the dense staggered configuration, the truncation effect is even stronger: any deviation from rigid two-site translation causes the support to spread into an undoped region, where the component is removed. 
Thus, in the block-staggered case as well, nonzero return cycles arise solely within the localized, rigidly translating Pauli subspace identified above. 
Viewed from this perspective, upper bounds on the gaps of block-staggered circuits are effectively controlled by the support size of localized Pauli strings that can form return cycles. 

To quantify this effective cycle support numerically, we proceed as follows. 
For each fixed block length $k$ and candidate cutoff weight $w$, we scan over all localized Pauli strings of maximum weight $w$ with contiguous support and evolve each of them under the Floquet map. After every step we truncate back to the subspace of strings with weight at most $w$, so that the dynamics is confined to the truncated weight-$w$ subspace.
We then check whether, for the given pair $(k,w)$, there exists at least one initial localized string whose support remains contiguous at every truncated step and is, after some finite number of steps, mapped back to the same Pauli configuration within an equivalent local block (up to an overall position translation). 
The smallest cutoff $w$ for which such a cycle occurs defines the effective cycle support, and the corresponding strings generate nonzero return cycles, since the subsequent Floquet evolution then repeats periodically.

\begin{table}[!t]
\label{tab:block_staggered}
\begin{ruledtabular}
\begin{tabular}{c cccc cccc}
      & \multicolumn{4}{c}{Odd \(k\)} & \multicolumn{4}{c}{Even \(k\)} \\
\cline{2-5}\cline{6-9}
\noalign{\vskip 3pt}
\(k\)      & 1 & 3 & 5 & 7 & 2 & 4 & 6 & 8 \\
\(w_\ast\) & 2 & 5 & 8 & 9 & 5 & 7 & 7 & 11 \\
\end{tabular}
\end{ruledtabular}
\caption{
{\bf Minimal truncation weight} $w_\ast$ required to obtain a localized return cycle in a block-staggered pattern with block length $k$, for which the doping density is $1/(k+1)$.}
\end{table}
Table~\ref{tab:block_staggered} shows, for each block length 
$k$, the minimal truncation weight $w_\ast$ at which our procedure identifies a localized return cycle in the block-staggered pattern, with representative cycles provided in Appendix~\ref{app:struc_block_staggered}.
Over all system sizes accessible numerically, $w_\ast$ is essentially independent of $N$ and grows at most linearly with $k$.

For $k=1$ the block-staggered pattern reduces to the staggered doped structure analyzed in Sec.~\ref{subsec:dense_dop}, and the localized return cycle identified above coincides with an exact eigenoperator basis for the slowest nontrivial mode. In this case the Liouvillian gap in the strong dissipation limit is given exactly by Eq.~\eqref{eq:stag_lgap_0}, with the coefficient multiplying $\gamma$ equal to $w_\ast=2$. 
For $k \ge 2$ the localized cycles produced by our construction are no longer exact eigenoperators of the Liouvillian gap, but, as in the dense-doping regime, they still yield an upper bound on the slowest decay rates: the presence of a localized return cycle with truncation weight $w_\ast$ implies that the prefactor of $\gamma$ in the Liouvillian gap is bounded from above by a finite function of $w_\ast$ that remains independent of the system size $N$. 

Taken together, these results show that regular block-staggered doping retains a system-size-independent relaxation scale for any fixed block length \(k\).
Moreover, our construction yields an upper bound in which the dissipative slope grows at most linearly with \(k\).
Since \(p_\mathrm{h}=1/(k+1)\), this structure can realize sparse doping while still maintaining a finite gap in the thermodynamic limit, suggesting that system-size-independent relaxation persists as long as \(k\) remains finite.

\section{Discussion} \label{sec:discussion}

We studied dissipative Haar-doped Floquet circuits (DHFC) built from an \textit{i}SWAP-class two-qubit Clifford gate, taking the fixed-gate realization with $i\mathrm{SWAP}(H\!\otimes\! H)$ as a representative example.
Within this setting, we analyzed the weak-dissipation singularity of the Liouvillian gap, focusing on how it evolves as the Haar-doping density is varied.
We now discuss the main physical implications of these results.

First, our results indicate, according to the dissipative diagnostic, that the undoped circuit realizes the strongest intrinsic relaxation in our setting. For any fixed $\gamma>0$, the Liouvillian gap grows linearly with $N$, so relaxation to the steady state accelerates as the system size increases. The mechanism is simple. The unitary dynamics rapidly spreads typical local Pauli operators to extensive weight $w(t)={\cal O}(N)$ on an ${\cal O}(N)$ timescale~\cite{Spreading2018}. Once operators are extensive, local noise of strength $\gamma$ damps them at a rate $\sim\gamma w$, i.e.\ $\sim\gamma N$. In the absence of conserved densities or other slow channels, there is no hydrodynamic mechanism that can produce small decay rates, and the Liouvillian spectrum develops a parametrically large gap.

Second, adding Haar random gates, contrary to our initial intuition based on simulation complexity, weakens intrinsic relaxation, in that the gap no longer grows with system size upon doping.
At the level of Pauli dynamics, local randomness strongly mixes the single-site Pauli components, continually regenerating low-weight Pauli string contributions.
Under repeated Floquet steps, this regeneration induces return cycles supported entirely within the low-weight sector, each confined to a small contiguous block that translates along the chain, a mechanism absent in the fine-tuned undoped circuit.
These low-weight return cycles dominate the slowest nontrivial eigenmodes and set the intrinsic relaxation rate, thereby rendering the gap finite in the thermodynamic limit.
Our analytic results precisely show that this crossover from unbounded growth of the gap with system size to a finite Liouvillian gap sets in once the doping density is nonvanishing.
Notably, this crossover sets in when the number of doped sites scales linearly with system size, the same regime in which scrambling diagnostics become Haar-like (e.g., OTOCs)~\cite{leone2021quantum, magni2025quantum} and higher $t$-moments approach Haar values~\cite{haferkamp2022random, harrow2023approximate}.
This suggests that relaxation- and scrambling-based diagnostics become nontrivial on parametrically similar doping scales.

A useful benchmark for placing these results is the self-dual kicked Ising (SDKI) Floquet circuit~\cite{bertini2018exact}: in this model, chaotic behavior has been characterized using standard unitary chaos diagnostics~\cite{flack2020statistics,bertini2019entanglement,bertini2020operator}, and ergodic/mixing behavior has also been discussed via dynamical correlation functions of single-site Pauli operators~\cite{bertini2019exact}.
Our DHFC construction can be viewed as an SDKI-like model augmented in two directions: site-resolved Haar doping and explicit local dissipation (so the Liouvillian gap is well-defined). This extension lets us tune the spatial doping structure and determine the doping-density scale at which the singularity crosses over, thereby separating distinct chaotic regimes in this diagnostic. Importantly, this conclusion holds both for quenched Haar disorder and for Haar rotations resampled each Floquet period.

Taken together, our finding of a finite-gap onset at linear doping and prior results on Haar-like scrambling suggest parametrically similar scales for relaxation- and scrambling-based diagnostics.
Interpreting this comparison, however, requires some care.
Most existing results are established in circuit ensembles with a different microscopic architecture, typically at greater depth or in different geometries than the shallow one-dimensional setting analyzed here.
Bringing these diagnostics into a single controlled framework therefore calls for extending the present model beyond one dimension and beyond the shallow-depth setting.
An important direction for future work is to explore whether the finite-gap scale correlates with approximate unitary $t$-design formation, thereby clarifying how these two diagnostics are related.
Ultimately, establishing such links between complementary diagnostics may help develop a more unified picture of quantum chaos beyond diagnostic-specific characterizations.

\acknowledgments
We thank Zi-Wen Liu, Lucas Sa, Jaewon Kim, and Sagar Vijay for fruitful discussions.
The work is supported by the faculty startup grant at the University of Illinois, Urbana-Champaign and the IBM-Illinois Discovery
Accelerator Institute. 
H.E.K. acknowledges support by the education and training program of the Quantum Information Research Support Center, funded through the National research foundation of Korea (NRF) by the Ministry of science and ICT (MSIT) of the Korean government(No. RS-2023-NR057243 and No. RS-2024-00432214).

\bibliography{ref_main}% Produces the bibliography via BibTeX.

@article{Spreading2018,
  title = {Operator Spreading in Random Unitary Circuits},
  author = {Nahum, Adam and Vijay, Sagar and Haah, Jeongwan},
  journal = {Phys. Rev. X},
  volume = {8},
  issue = {2},
  pages = {021014},
  numpages = {30},
  year = {2018},
  month = {Apr},
  publisher = {American Physical Society},
  doi = {10.1103/PhysRevX.8.021014},
  url = {https://link.aps.org/doi/10.1103/PhysRevX.8.021014}
}

@incollection{haake1991quantum,
  title={Quantum signatures of chaos},
  author={Haake, Fritz},
  booktitle={Quantum coherence in mesoscopic systems},
  pages={583--595},
  year={1991},
  publisher={Springer}
}

@article{bohigas1984characterization,
  title = {Characterization of Chaotic Quantum Spectra and Universality of Level Fluctuation Laws},
  author = {Bohigas, O. and Giannoni, M. J. and Schmit, C.},
  journal = {Phys. Rev. Lett.},
  volume = {52},
  issue = {1},
  pages = {1--4},
  numpages = {0},
  year = {1984},
  month = {Jan},
  publisher = {American Physical Society},
  doi = {10.1103/PhysRevLett.52.1},
  url = {https://link.aps.org/doi/10.1103/PhysRevLett.52.1}
}

@article{guhr1998random,
  title={Random-matrix theories in quantum physics: common concepts},
  author={Guhr, Thomas and M{\"u}ller--Groeling, Axel and Weidenm{\"u}ller, Hans A},
  journal={Physics Reports},
  volume={299},
  number={4-6},
  pages={189--425},
  year={1998},
  publisher={Elsevier},
  url={https://doi.org/10.1016/S0370-1573(97)00088-4}
}

@article{d2016quantum,
  title={From quantum chaos and eigenstate thermalization to statistical mechanics and thermodynamics},
  author={D'Alessio, Luca and Kafri, Yariv and Polkovnikov, Anatoli and Rigol, Marcos},
  journal={Advances in Physics},
  volume={65},
  number={3},
  pages={239--362},
  year={2016},
  publisher={Taylor \& Francis},
  url={https://doi.org/10.1080/00018732.2016.1198134}
}

@article{deutsch2018eigenstate,
  title={Eigenstate thermalization hypothesis},
  author={Deutsch, Joshua M},
  journal={Reports on Progress in Physics},
  volume={81},
  number={8},
  pages={082001},
  year={2018},
  publisher={IOP Publishing},
  url = {https://doi.org/10.1088/1361-6633/aac9f1}
}

@article{eisert2015quantum,
  title={Quantum many-body systems out of equilibrium},
  author={Eisert, Jens and Friesdorf, Mathis and Gogolin, Christian},
  journal={Nature Physics},
  volume={11},
  number={2},
  pages={124--130},
  year={2015},
  publisher={Nature Publishing Group UK London},
  url={https://doi.org/10.1038/nphys3215}
}

@article{leblond2021universality,
  title = {Universality in the onset of quantum chaos in many-body systems},
  author = {LeBlond, Tyler and Sels, Dries and Polkovnikov, Anatoli and Rigol, Marcos},
  journal = {Phys. Rev. B},
  volume = {104},
  issue = {20},
  pages = {L201117},
  numpages = {7},
  year = {2021},
  month = {Nov},
  publisher = {American Physical Society},
  doi = {10.1103/PhysRevB.104.L201117},
  url = {https://link.aps.org/doi/10.1103/PhysRevB.104.L201117}
}

@article{roberts2017chaos,
  title={Chaos and complexity by design},
  author={Roberts, Daniel A and Yoshida, Beni},
  journal={Journal of High Energy Physics},
  volume={2017},
  number={4},
  pages={1--64},
  year={2017},
  publisher={Springer},
  url={https://doi.org/10.1007/JHEP04(2017)121}
}

@article{harrow2021separation,
  title = {Separation of Out-Of-Time-Ordered Correlation and Entanglement},
  author = {Harrow, Aram W. and Kong, Linghang and Liu, Zi-Wen and Mehraban, Saeed and Shor, Peter W.},
  journal = {PRX Quantum},
  volume = {2},
  issue = {2},
  pages = {020339},
  numpages = {12},
  year = {2021},
  month = {Jun},
  publisher = {American Physical Society},
  doi = {10.1103/PRXQuantum.2.020339},
  url = {https://link.aps.org/doi/10.1103/PRXQuantum.2.020339}
}

@article{hosur2016chaos,
  title={Chaos in quantum channels},
  author={Hosur, Pavan and Qi, Xiao-Liang and Roberts, Daniel A and Yoshida, Beni},
  journal={Journal of High Energy Physics},
  volume={2016},
  number={2},
  pages={1--49},
  year={2016},
  publisher={Springer},
  url={https://doi.org/10.1007/JHEP02(2016)004}
}

@article{nahum2018operator,
  title = {Operator Spreading in Random Unitary Circuits},
  author = {Nahum, Adam and Vijay, Sagar and Haah, Jeongwan},
  journal = {Phys. Rev. X},
  volume = {8},
  issue = {2},
  pages = {021014},
  numpages = {30},
  year = {2018},
  month = {Apr},
  publisher = {American Physical Society},
  doi = {10.1103/PhysRevX.8.021014},
  url = {https://link.aps.org/doi/10.1103/PhysRevX.8.021014}
}

@article{nahum2017quantum,
  title = {Quantum Entanglement Growth under Random Unitary Dynamics},
  author = {Nahum, Adam and Ruhman, Jonathan and Vijay, Sagar and Haah, Jeongwan},
  journal = {Phys. Rev. X},
  volume = {7},
  issue = {3},
  pages = {031016},
  numpages = {30},
  year = {2017},
  month = {Jul},
  publisher = {American Physical Society},
  doi = {10.1103/PhysRevX.7.031016},
  url = {https://link.aps.org/doi/10.1103/PhysRevX.7.031016}
}

@article{prosen2007efficiency,
  title = {Is the efficiency of classical simulations of quantum dynamics related to integrability?},
  author = {Prosen, Toma\ifmmode \check{z}\else \v{z}\fi{} and \ifmmode \check{Z}\else \v{Z}\fi{}nidari\ifmmode \check{c}\else \v{c}\fi{}, Marko},
  journal = {Phys. Rev. E},
  volume = {75},
  issue = {1},
  pages = {015202},
  numpages = {4},
  year = {2007},
  month = {Jan},
  publisher = {American Physical Society},
  doi = {10.1103/PhysRevE.75.015202},
  url = {https://link.aps.org/doi/10.1103/PhysRevE.75.015202}
}

@article{dowling2023scrambling,
  title = {Scrambling Is Necessary but Not Sufficient for Chaos},
  author = {Dowling, Neil and Kos, Pavel and Modi, Kavan},
  journal = {Phys. Rev. Lett.},
  volume = {131},
  issue = {18},
  pages = {180403},
  numpages = {6},
  year = {2023},
  month = {Nov},
  publisher = {American Physical Society},
  doi = {10.1103/PhysRevLett.131.180403},
  url = {https://link.aps.org/doi/10.1103/PhysRevLett.131.180403}
}

@article{von2018operator,
  title = {Operator Hydrodynamics, OTOCs, and Entanglement Growth in Systems without Conservation Laws},
  author = {von Keyserlingk, C. W. and Rakovszky, Tibor and Pollmann, Frank and Sondhi, S. L.},
  journal = {Phys. Rev. X},
  volume = {8},
  issue = {2},
  pages = {021013},
  numpages = {19},
  year = {2018},
  month = {Apr},
  publisher = {American Physical Society},
  doi = {10.1103/PhysRevX.8.021013},
  url = {https://link.aps.org/doi/10.1103/PhysRevX.8.021013}
}

@article{turner2018weak,
  title={Weak ergodicity breaking from quantum many-body scars},
  author={Turner, Christopher J and Michailidis, Alexios A and Abanin, Dmitry A and Serbyn, Maksym and Papi{\'c}, Zlatko},
  journal={Nature Physics},
  volume={14},
  number={7},
  pages={745--749},
  year={2018},
  publisher={Nature Publishing Group UK London},
  url = {https://doi.org/10.1038/s41567-018-0137-5}
}

@article{pizzi2025genuine,
  title={Genuine quantum scars in many-body spin systems},
  author={Pizzi, Andrea and Kwan, Long-Hei and Evrard, Bertrand and Dag, Ceren B and Knolle, Johannes},
  journal={Nature Communications},
  volume={16},
  number={1},
  pages={6722},
  year={2025},
  publisher={Nature Publishing Group UK London},
  url={https://doi.org/10.1038/s41467-025-61765-3}
}

@article{zhou2020single,
	title={{Single T gate in a Clifford circuit drives transition to universal entanglement spectrum statistics}},
	author={Shiyu Zhou and Zhi-Cheng Yang and Alioscia Hamma and Claudio Chamon},
	journal={SciPost Phys.},
	volume={9},
	pages={087},
	year={2020},
	publisher={SciPost},
	doi={10.21468/SciPostPhys.9.6.087},
	url={https://scipost.org/10.21468/SciPostPhys.9.6.087}
}

@article{leone2021quantum,
  title={Quantum chaos is quantum},
  author={Leone, Lorenzo and Oliviero, Salvatore FE and Zhou, You and Hamma, Alioscia},
  journal={Quantum},
  volume={5},
  pages={453},
  year={2021},
  publisher={Verein zur F{\"o}rderung des Open Access Publizierens in den Quantenwissenschaften},
  url={https://doi.org/10.22331/q-2021-05-04-453}
}

@article{magni2025quantum,
  title={Quantum Complexity and Chaos in Many-Qudit Doped Clifford Circuits},
  author={Magni, Beatrice and Turkeshi, Xhek},
  journal={arXiv preprint arXiv:2506.02127},
  year={2025},
  url={https://doi.org/10.48550/arXiv.2506.02127}
}

@article{mori2024liouvillian,
  title = {Liouvillian-gap analysis of open quantum many-body systems in the weak dissipation limit},
  author = {Mori, Takashi},
  journal = {Phys. Rev. B},
  volume = {109},
  issue = {6},
  pages = {064311},
  numpages = {10},
  year = {2024},
  month = {Feb},
  publisher = {American Physical Society},
  doi = {10.1103/PhysRevB.109.064311},
  url = {https://link.aps.org/doi/10.1103/PhysRevB.109.064311}
}

@article{zhang2024thermalization,
  title={Thermalization rates and quantum Ruelle-Pollicott resonances: insights from operator hydrodynamics},
  author={Zhang, Carolyn and Nie, Laimei and von Keyserlingk, Curt},
  journal={arXiv preprint arXiv:2409.17251},
  year={2024},
  url={https://doi.org/10.48550/arXiv.2409.17251}
}

@article{jacoby2025spectral,
  title = {Spectral gaps of local quantum channels in the weak-dissipation limit},
  author = {Jacoby, J. Alexander and Huse, David A. and Gopalakrishnan, Sarang},
  journal = {Phys. Rev. B},
  volume = {111},
  issue = {10},
  pages = {104303},
  numpages = {8},
  year = {2025},
  month = {Mar},
  publisher = {American Physical Society},
  doi = {10.1103/PhysRevB.111.104303},
  url = {https://link.aps.org/doi/10.1103/PhysRevB.111.104303}
}

@article{yoshimura2024robustness,
  title={Robustness of quantum chaos and anomalous relaxation in open quantum circuits},
  author={Yoshimura, Takato and S{\'a}, Lucas},
  journal={Nature Communications},
  volume={15},
  number={1},
  pages={9808},
  year={2024},
  publisher={Nature Publishing Group UK London},
  url={https://doi.org/10.1038/s41467-024-54164-7}
}

@article{yoshimura2025theory,
  title = {Theory of irreversibility in quantum many-body systems},
  author = {Yoshimura, Takato and S\'a, Lucas},
  journal = {Phys. Rev. E},
  volume = {111},
  issue = {6},
  pages = {064135},
  numpages = {22},
  year = {2025},
  month = {Jun},
  publisher = {American Physical Society},
  doi = {10.1103/82f6-qdyd},
  url = {https://link.aps.org/doi/10.1103/82f6-qdyd}
}

@article{garcia2023keldysh,
  title = {Keldysh wormholes and anomalous relaxation in the dissipative Sachdev-Ye-Kitaev model},
  author = {Garc\'{\i}a-Garc\'{\i}a, Antonio M. and S\'a, Lucas and Verbaarschot, Jacobus J. M. and Zheng, Jie Ping},
  journal = {Phys. Rev. D},
  volume = {107},
  issue = {10},
  pages = {106006},
  numpages = {27},
  year = {2023},
  month = {May},
  publisher = {American Physical Society},
  doi = {10.1103/PhysRevD.107.106006},
  url = {https://link.aps.org/doi/10.1103/PhysRevD.107.106006}
}

@article{vznidarivc2024momentum,
  title = {Momentum-dependent quantum Ruelle-Pollicott resonances in translationally invariant many-body systems},
  author = {\ifmmode \check{Z}\else \v{Z}\fi{}nidari\ifmmode \check{c}\else \v{c}\fi{}, Marko},
  journal = {Phys. Rev. E},
  volume = {110},
  issue = {5},
  pages = {054204},
  numpages = {13},
  year = {2024},
  month = {Nov},
  publisher = {American Physical Society},
  doi = {10.1103/PhysRevE.110.054204},
  url = {https://link.aps.org/doi/10.1103/PhysRevE.110.054204}
}

@article{liao2022emergence,
  title = {Emergence of many-body quantum chaos via spontaneous breaking of unitarity},
  author = {Liao, Yunxiang and Galitski, Victor},
  journal = {Phys. Rev. B},
  volume = {105},
  issue = {14},
  pages = {L140202},
  numpages = {6},
  year = {2022},
  month = {Apr},
  publisher = {American Physical Society},
  doi = {10.1103/PhysRevB.105.L140202},
  url = {https://link.aps.org/doi/10.1103/PhysRevB.105.L140202}
}

@article{duh2025ruelle,
  title={Ruelle-Pollicott resonances of diffusive U (1)-invariant qubit circuits},
  author={Duh, Urban and {\v{Z}}nidari{\v{c}}, Marko},
  journal={arXiv preprint arXiv:2506.24097},
  year={2025},
  url={https://doi.org/10.48550/arXiv.2506.24097}
}

@article{grier2022classification,
  title={The classification of Clifford gates over qubits},
  author={Grier, Daniel and Schaeffer, Luke},
  journal={Quantum},
  volume={6},
  pages={734},
  year={2022},
  publisher={Verein zur F{\"o}rderung des Open Access Publizierens in den Quantenwissenschaften},
  url={https://doi.org/10.22331/q-2022-06-13-734}
}

@article{kong2024convergence,
  title={Convergence efficiency of quantum gates and circuits},
  author={Kong, Linghang and Li, Zimu and Liu, Zi-Wen},
  journal={arXiv preprint arXiv:2411.04898},
  year={2024},
  url={https://doi.org/10.48550/arXiv.2411.04898}
}

@article{kovacs2024operator,
  title={Operator space fragmentation in perturbed Floquet-Clifford circuits},
  author={Kov{\'a}cs, Marcell D and Turner, Christopher J and Masanes, Lluis and Pal, Arijeet},
  journal={arXiv preprint arXiv:2408.01545},
  year={2024},
  url={https://doi.org/10.48550/arXiv.2408.01545}
}

@book{nielsen2010quantum,
  title={Quantum computation and quantum information},
  author={Nielsen, Michael A and Chuang, Isaac L},
  year={2010},
  publisher={Cambridge university press}
}

@book{varga2011gervsgorin,
  title={Ger{\v{s}}gorin and his circles},
  author={Varga, Richard S},
  volume={36},
  year={2011},
  publisher={Springer Science \& Business Media}
}

@article{bertini2019exact,
  title = {Exact Correlation Functions for Dual-Unitary Lattice Models in $1+1$ Dimensions},
  author = {Bertini, Bruno and Kos, Pavel and Prosen, Toma\ifmmode \check{z}\else \v{z}\fi{}},
  journal = {Phys. Rev. Lett.},
  volume = {123},
  issue = {21},
  pages = {210601},
  numpages = {6},
  year = {2019},
  month = {Nov},
  publisher = {American Physical Society},
  doi = {10.1103/PhysRevLett.123.210601},
  url = {https://link.aps.org/doi/10.1103/PhysRevLett.123.210601}
}

@article{bertini2018exact,
  title = {Exact Spectral Form Factor in a Minimal Model of Many-Body Quantum Chaos},
  author = {Bertini, Bruno and Kos, Pavel and Prosen, Toma\ifmmode \check{z}\else \v{z}\fi{}},
  journal = {Phys. Rev. Lett.},
  volume = {121},
  issue = {26},
  pages = {264101},
  numpages = {6},
  year = {2018},
  month = {Dec},
  publisher = {American Physical Society},
  doi = {10.1103/PhysRevLett.121.264101},
  url = {https://link.aps.org/doi/10.1103/PhysRevLett.121.264101}
}

@article{flack2020statistics,
  title = {Statistics of the spectral form factor in the self-dual kicked Ising model},
  author = {Flack, Ana and Bertini, Bruno and Prosen, Toma\ifmmode \check{z}\else \v{z}\fi{}},
  journal = {Phys. Rev. Res.},
  volume = {2},
  issue = {4},
  pages = {043403},
  numpages = {14},
  year = {2020},
  month = {Dec},
  publisher = {American Physical Society},
  doi = {10.1103/PhysRevResearch.2.043403},
  url = {https://link.aps.org/doi/10.1103/PhysRevResearch.2.043403}
}

@article{bertini2019entanglement,
  title = {Entanglement Spreading in a Minimal Model of Maximal Many-Body Quantum Chaos},
  author = {Bertini, Bruno and Kos, Pavel and Prosen, Toma\ifmmode \check{z}\else \v{z}\fi{}},
  journal = {Phys. Rev. X},
  volume = {9},
  issue = {2},
  pages = {021033},
  numpages = {27},
  year = {2019},
  month = {May},
  publisher = {American Physical Society},
  doi = {10.1103/PhysRevX.9.021033},
  url = {https://link.aps.org/doi/10.1103/PhysRevX.9.021033}
}

@article{bertini2020operator,
  title={Operator entanglement in local quantum circuits I: Chaotic dual-unitary circuits},
  author={Bertini, Bruno and Kos, Pavel and Prosen, Toma{\v{z}}},
  journal={SciPost Physics},
  volume={8},
  number={4},
  pages={067},
  year={2020},
  url={https://doi.org/10.21468/SciPostPhys.8.4.067}
}

@article{Haferkamp2022random,
  doi = {10.22331/q-2022-09-08-795},
  url = {https://doi.org/10.22331/q-2022-09-08-795},
  title = {Random quantum circuits are approximate unitary {$t$}-designs in depth {$O\left(nt^{5+o(1)}\right)$}},
  author = {Haferkamp, Jonas},
  journal = {{Quantum}},
  issn = {2521-327X},
  publisher = {{Verein zur F{\"{o}}rderung des Open Access Publizierens in den Quantenwissenschaften}},
  volume = {6},
  pages = {795},
  month = sep,
  year = {2022}
}

@article{harrow2023approximate,
  title={Approximate unitary t-designs by short random quantum circuits using nearest-neighbor and long-range gates},
  author={Harrow, Aram W and Mehraban, Saeed},
  journal={Communications in Mathematical Physics},
  volume={401},
  number={2},
  pages={1531--1626},
  year={2023},
  publisher={Springer},
  url ={https://doi.org/10.1007/s00220-023-04675-z}
}

@article{deutsch1991quantum,
  title = {Quantum statistical mechanics in a closed system},
  author = {Deutsch, J. M.},
  journal = {Phys. Rev. A},
  volume = {43},
  issue = {4},
  pages = {2046--2049},
  numpages = {0},
  year = {1991},
  month = {Feb},
  publisher = {American Physical Society},
  doi = {10.1103/PhysRevA.43.2046},
  url = {https://link.aps.org/doi/10.1103/PhysRevA.43.2046}
}

@article{srednicki1994chaos,
  title = {Chaos and quantum thermalization},
  author = {Srednicki, Mark},
  journal = {Phys. Rev. E},
  volume = {50},
  issue = {2},
  pages = {888--901},
  numpages = {0},
  year = {1994},
  month = {Aug},
  publisher = {American Physical Society},
  doi = {10.1103/PhysRevE.50.888},
  url = {https://link.aps.org/doi/10.1103/PhysRevE.50.888}
}

@article{rigol2008thermalization,
  title={Thermalization and its mechanism for generic isolated quantum systems},
  author={Rigol, Marcos and Dunjko, Vanja and Olshanii, Maxim},
  journal={Nature},
  volume={452},
  number={7189},
  pages={854--858},
  year={2008},
  publisher={Nature Publishing Group UK London},
  url={https://doi.org/10.1038/nature06838}
}

@article{mori2018thermalization,
  title={Thermalization and prethermalization in isolated quantum systems: a theoretical overview},
  author={Mori, Takashi and Ikeda, Tatsuhiko N and Kaminishi, Eriko and Ueda, Masahito},
  journal={Journal of Physics B: Atomic, Molecular and Optical Physics},
  volume={51},
  number={11},
  pages={112001},
  year={2018},
  publisher={IOP Publishing},
  url = {https://doi.org/10.1088/1361-6455/aabcdf}
}

\clearpage %new page

\appendix
\onecolumngrid % make as one-column

\section{Derivation of the depolarizing channel action}
\label{app:dep}

We start from the standard Markovian depolarization master equation on a single qubit,
\begin{align}
\dot\rho(t)=-\gamma_0\Bigl(\rho(t)-\frac{I}{2}\Bigr).
\label{eq:dep_master}
\end{align}
Its solution is
\begin{align}
\rho(t)=\frac{I}{2}+e^{-\gamma_0 t}\Bigl(\rho(0)-\frac{I}{2}\Bigr).
\label{eq:dep_solution_rho}
\end{align}
We write this evolution as a quantum channel by defining $\mathcal{N}^{\mathrm{dep}}_{1}$ through
$\rho(t)=\mathcal{N}^{\mathrm{dep}}_{1}(\rho(0))$.
Using the following identity
\begin{align}
\frac{I}{2}=\frac14\bigl(\rho+X\rho X+Y\rho Y+Z\rho Z\bigr),
\label{eq:pauli_twirl}
\end{align}
and substituting it into Eq.~\eqref{eq:dep_solution_rho}, one immediately finds that $\mathcal{N}^{\mathrm{dep}}_{1}$
acts diagonally in the Pauli basis:
\begin{align}
\mathcal{N}^{\mathrm{dep}}_{1}(I)=I,
\qquad
\mathcal{N}^{\mathrm{dep}}_{1}(P)=e^{-\gamma}P,
\quad P\in\{X,Y,Z\},
\label{eq:dep_pauli_action_compact}
\end{align}
where $\gamma:=\gamma_0 t$ is dissipation strength.

For an $N$-qubit chain, the full noise layer acts independently on each site,
$\mathcal{N}=\bigotimes_{j=1}^N \mathcal{N}_1^{(j)}$.
Applying Eq.~\eqref{eq:dep_pauli_action_compact} sitewise yields
\begin{align}
\mathcal{N}(S)
=\bigotimes_{j=1}^N \mathcal{N}_1^{(j)}(\sigma_{\mu_j}^j)
=\Bigl(\prod_{j:\mu_j\neq 0} e^{-\gamma}\Bigr) S
=e^{-\gamma\,w(S)}\,S,
\label{eq:depol_diag_app}
\end{align}
which is Eq.~\eqref{eq:depol_diag} in the main text.

\section{Haar average of the logarithmic observable}
\label{app:haar-average}

For any Haar-random single-qubit unitary \(U \in \mathrm{SU}(2)\), we define the Pauli transfer matrix elements
\begin{equation}
    U_{ij} := \frac{1}{2} \mathrm{Tr}\left[ U \sigma_i U^\dagger \sigma_j \right], \qquad i,j \in \{x,y,z\},
    \label{eq:app_uij}
\end{equation}
which describe the action of the unitary channel \(\mathcal{U}(\rho) = U\rho U^\dagger\) in the Pauli basis. The \(3 \times 3\) real matrix \(U_{ij}\) corresponds to the adjoint representation of \(U\), and forms a special orthogonal matrix in \(\mathrm{SO}(3)\).

Under Haar measure, the adjoint representation of \(U\) becomes Haar-random in \(\mathrm{SO}(3)\). In this ensemble, each element \(U_{ij}\) has the same marginal distribution as a coordinate of a uniformly random unit vector on the Bloch sphere. In particular, for all \(i,j \in \{x,y,z\}\),
\begin{equation}
    \mathbb{E}_{U \sim \mathrm{Haar}} \left[ \log \left| U_{ij} \right| \right] = -1.
    \label{eq:app_log_uij}
\end{equation}
This follows from the fact that each \(U_{ij}\) is distributed uniformly in \([-1,1]\), and hence
\begin{equation}
\label{eq:app_final}
    \mathbb{E} \left[ \log |U_{ij}| \right]
    = \int_{-1}^{1} \frac{1}{2} \log|x|\, dx
    = \int_0^1 \log x\, dx = -1.
\end{equation}
Since $U_{ij} = 0$ occurs with zero probability under Haar measure, the logarithmic expectation is well-defined despite the singularity.
\\

As required in the derivation of the Liouvillian gap upper bound in Appendix~\ref{app:Lgap_upp_bound}, we further evaluate the Haar average of a logarithmic observable defined by a bilinear combination of Pauli transfer matrix elements:
\begin{equation}
    \mathbb{E}_{U,V \sim \mathrm{Haar}} \left[ \log \left| U_{33} V_{22} + U_{32} V_{12} \right| \right],
\end{equation}
where \(U, V \in \mathrm{SU}(2)\) are independently drawn Haar-random unitaries, and \(U_{ij} := \frac{1}{2} \mathrm{Tr}[U \sigma_i U^\dagger \sigma_j]\) as in Eq.~\eqref{eq:app_uij}.

We interpret the expression as a scalar product between two-dimensional projections:
\begin{equation}
    U_{33} V_{22} + U_{32} V_{12} = (U_{32},\, U_{33}) \cdot (V_{12},\, V_{22}) =: \vec{u} \cdot \vec{v}.
\end{equation}
The vector \(\vec{u} = (U_{32}, U_{33})\) consists of two entries from the third row of \(U_{ij}\), and \(\vec{v} = (V_{12}, V_{22})\) consists of two entries from the second column of \(V_{ij}\). Under Haar measure, these correspond to random two-dimensional projections of unit vectors on the Bloch sphere:
\begin{align}
    \|\vec{u}\| &= \sqrt{1 - U_{31}^2}, \\
    \|\vec{v}\| &= \sqrt{1 - V_{32}^2},
\end{align}
with independent angles between \(\vec{u}\) and \(\vec{v}\) uniformly distributed on the circle.

Hence, the logarithmic expectation separates as
\begin{align}
    \mathbb{E} \left[ \log \left| U_{33} V_{22} + U_{32} V_{12} \right| \right]
    &= \mathbb{E}[\log \|\vec{u}\|] + \mathbb{E}[\log \|\vec{v}\|] + \mathbb{E}[\log |\cos\theta|] \\
    &= (\log 2 - 1) + (\log 2 - 1) + (-\log 2) = \log 2 - 2.
\end{align}

We conclude:
\begin{equation}
    \mathbb{E}_{U,V \sim \mathrm{Haar}} \left[ \log \left| U_{33} V_{22} + U_{32} V_{12} \right| \right]
    = \log 2 - 2.
\end{equation}
This result holds for any bilinear expression of the form \(U_{ij}V_{kl} + U_{mn}V_{pq}\), provided that each pair \((U_{ij}, U_{mn})\) and \((V_{kl}, V_{pq})\) lies in a common row or column of the corresponding SO(3) transfer matrix.

\section{Single-qubit unitary channels as linear combinations of Clifford conjugations}
\label{app:clifford-decomp}

In this appendix, we show that any single-qubit unitary conjugation channel can be expressed as a linear combination of Clifford conjugation maps. This structural property justifies the use of discrete Clifford orbits in our analysis of operator spreading and eigenmode suppression, even when the circuit contains non-Clifford Haar-doped layers.
\\

A key observation is that for any fixed unitary $U \in \mathrm{SU}(2)$, the channel $\mathcal{U}(\cdot) := U (\cdot) U^\dagger$ defines a linear map on the space of single-qubit operators. When expressed in the Pauli basis $\{I, X, Y, Z\}$, this map takes a particularly simple form: the identity operator is mapped to itself, and there is no mixing between the identity and any other Pauli operators. That is, the identity component remains invariant, and the conjugation acts nontrivially only within the subspace spanned by $\{X, Y, Z\}$.
\\

The action of $\mathcal{U}$ on this traceless subspace can be viewed as an orthogonal transformation in $\mathbb{R}^3$, preserving both the norm and the linear structure of the Pauli vector components. In particular, the resulting $3 \times 3$ submatrix corresponds to a proper rotation and can be decomposed into a linear combination of permutation-like maps on the Pauli operators. A convenient partition of the nine possible transition elements within this subspace is given by the following three canonical permutation patterns:
\begin{itemize}
    \item[(i)] $X \mapsto X$, $Y \mapsto Y$, $Z \mapsto Z$,
    \item[(ii)] $X \mapsto Y$, $Y \mapsto Z$, $Z \mapsto X$,
    \item[(iii)] $X \mapsto Z$, $Z \mapsto Y$, $Y \mapsto X$.
\end{itemize}
Each of these patterns is realized, up to a sign, by conjugation under a single-qubit Clifford unitary: 
pattern (i) by the identity $I$, pattern (ii) by $SH$, and pattern (iii) by $HS$.
\\

For each of the three permutation patterns described above, we construct a family of maps by multiplying the entire transformation with one of the Pauli operators $X$, $Y$, or $Z$. These multiplications act as independent sign flips on the image of each Pauli operator under the given permutation, while preserving the permutation structure itself. The resulting three maps are linearly independent and span a three-dimensional real subspace of conjugation maps. Together, the spans associated with the three permutation patterns generate the full nine-dimensional space of single-qubit conjugation maps on traceless operators.
\\

Therefore, any single-qubit unitary channel can be expressed as a linear combination of Clifford conjugation maps as
\begin{align}
\mathcal{U}(\cdot) = U (\cdot) U^\dagger = \sum_{i=1}^{3} \left[ a_{1,i}\, \sigma_i (\cdot) \sigma_i + a_{2,i}\, HS\, \sigma_i (\cdot) \sigma_i (HS)^\dagger + a_{3,i}\, SH\, \sigma_i (\cdot) \sigma_i (SH)^\dagger \right],
\end{align}
for some real coefficients \( \{ a_{j,i} \in \mathbb{R} \} \).
This decomposition allows us to interpret arbitrary unitary conjugation as a finite linear combination over Clifford channels, and enables a fully discrete analysis of operator dynamics in Haar-doped circuits.

\section{Feingold Varga Theorem}
\label{app:Feingold_Varga}
\begin{lemma}
    Let
\begin{align}
K=\begin{pmatrix}A & B\\ C & D\end{pmatrix},\qquad
A\in\mathbb{C}^{n_A\times n_A},\;
B\in\mathbb{C}^{n_A\times n_D},\;
C\in\mathbb{C}^{n_D\times n_A},\;
D\in\mathbb{C}^{n_D\times n_D}.
\end{align}
Assume that \(A\) is nilpotent (i.e., every eigenvalue of \(A\) equals \(0\)).
Let \(\|\cdot\|\) denote any induced (operator) matrix norm (so that \(\|XY\|\le\|X\|\,\|Y\|\) for conformable \(X,Y\)).
Define the spectral radius of a square matrix \(X\) by
\begin{align}
\rho(X):=\max\{\,|\lambda|:\ \lambda\ \text{is an eigenvalue of }X\,\}.
\end{align}
Then the following upper bound holds:
\begin{align}
\rho(K)\ 
%\le\ \frac{\ \rho(D)+\sqrt{\rho(D)^2+4\,\|B\|\,\|C\|}\ }{2}\ 
\le\ \rho(D)+\sqrt{\|B\|\,\|C\|}.
\end{align}
\end{lemma}

\begin{proof}
Let $K=\begin{pmatrix}A & B\\ C & D\end{pmatrix}$ with $A$ nilpotent (hence every eigenvalue of $A$ equals $0$), and let $\|\cdot\|$ be any induced matrix norm. For $\tau>0$ define
\begin{align}
S_\tau=\operatorname{diag}(\tau I_{n_A},\,I_{n_D}),\qquad
K_\tau:=S_\tau^{-1}KS_\tau
=
\begin{pmatrix}
A & \tau^{-1}B\\[2pt]
\tau C & D
\end{pmatrix}.
\end{align}
Similarity preserves eigenvalues, so $\rho(K_\tau)=\rho(K)$.

By the block Gershgorin inclusion (Feingold-Varga \cite{varga2011gervsgorin}), if $\lambda$ is an eigenvalue of $K_\tau$, then either
\begin{align}
\mathrm{dist}\big(\lambda,\{\text{eigenvalues of }A\}\big)\le \|\tau^{-1}B\|
\quad\text{or}\quad
\mathrm{dist}\big(\lambda,\{\text{eigenvalues of }D\}\big)\le \|\tau C\|.
\end{align}
Since every eigenvalue of $A$ equals $0$, the first alternative gives $|\lambda|\le \tau^{-1}\|B\|$, while the second implies $|\lambda|\le \rho(D)+\tau\|C\|$. Hence, for all $\tau>0$,
\begin{align}
\rho(K)=\rho(K_\tau)\ \le\ f(\tau):=\max\!\left\{\frac{\|B\|}{\tau},\ \rho(D)+\tau\|C\|\right\}.
\end{align}
Let $a:=\|B\|$, $b:=\rho(D)$, $c:=\|C\|$. The two branches $a/\tau$ and $b+c\tau$ intersect at the positive solution of $a/\tau=b+c\tau$, namely
\begin{align}
\tau_\star=\frac{-b+\sqrt{\,b^2+4ac\,}}{2c},
\end{align}
and at this point
\begin{align}
\min_{\tau>0} f(\tau)=f(\tau_\star)=\frac{\,b+\sqrt{\,b^2+4ac\,}}{2}.
\end{align}
Therefore
\begin{align}
\rho(K)\ \le\ \frac{\ \rho(D)+\sqrt{\rho(D)^2+4\,\|B\|\,\|C\|}\ }{2}.
\end{align}
Finally, using $\sqrt{\rho(D)^2+4ac}\le \rho(D)+2\sqrt{ac}$ for $a,c\ge0$, we obtain the simpler bound
\begin{align}
\rho(K)\ \le\ \rho(D)+\sqrt{\|B\|\,\|C\|}.
\end{align}
\end{proof}

\section{Liouvillian Gap of Staggered Doping Structure}
\label{app:struc_staggered}

We provide an analytic derivation of the Liouvillian gap in the staggered doping configuration, where every even site is doped (\(n_h = N/2\)) and every odd site is undoped.
We compute the asymptotic value of the gap in the thermodynamic limit \(N \to \infty\).

Numerical calculations show that for \(N \geq 6\), the only nonzero eigenoperators within the weight-2 truncated subspace are supported on the set \(\{Y_{2i}X_{2i+2} \mid i = 1,2,\ldots,N/2\}\).
This is because one Floquet circuit step maps each operator \(Y_{2i}X_{2i+2}\) to a linear combination of \(Y_{2i-2}X_{2i}\) and higher-weight terms, which are projected out by the weight-2 truncation.
\begin{figure}
    \centering
    \includegraphics[width=0.5\linewidth]{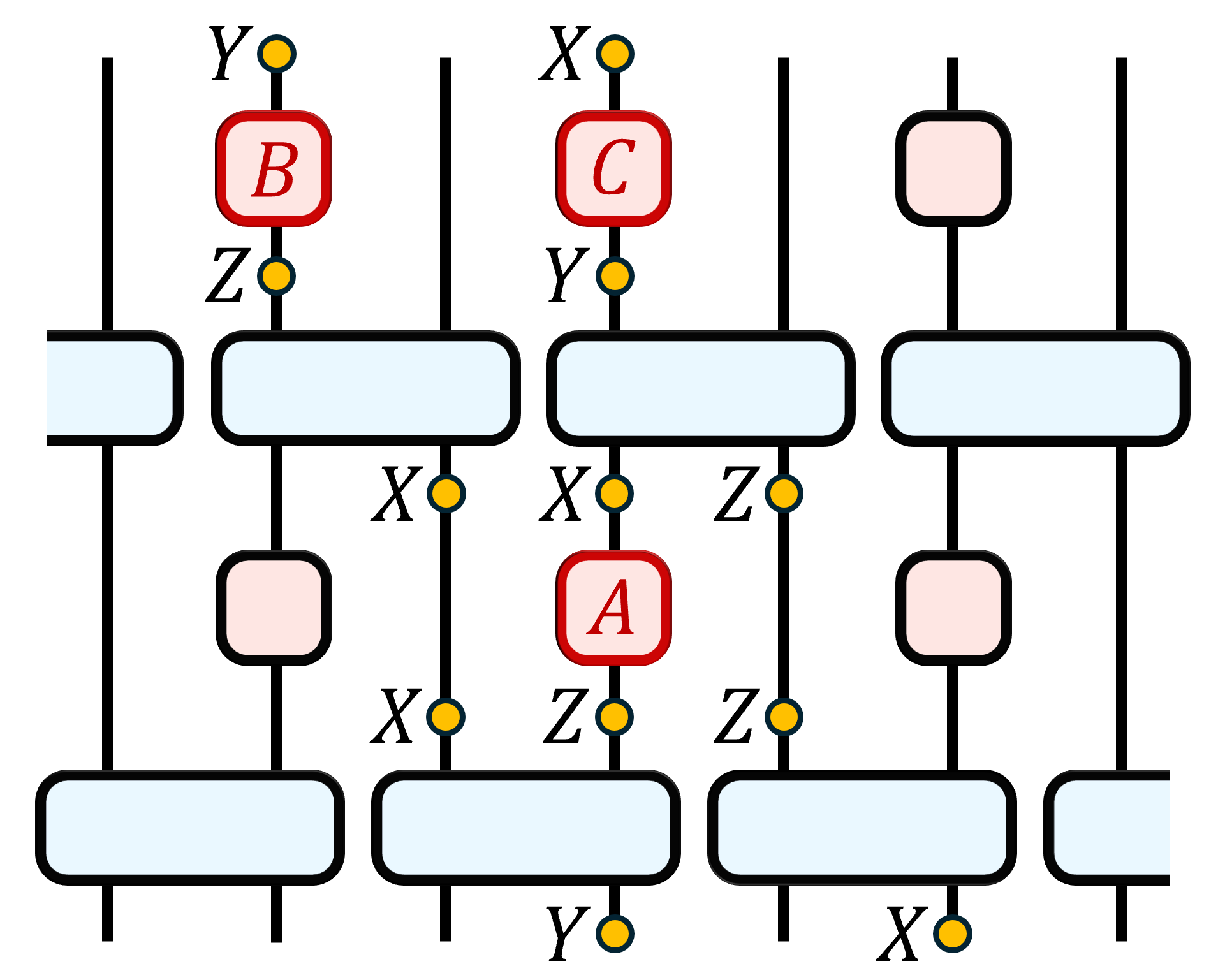}
    \caption{
    Weight-$2$ return cycle for staggered doping.
    Yellow dots mark nontrivial Pauli support (identities elsewhere).
    Under one Floquet period, $Y_{2i}X_{2i+2}$ maps to $Y_{2i-2}X_{2i}$ up to higher-weight terms discarded by the weight-$2$ truncation, yielding a two-site translation and a cycle.
    The three single-qubit Haar rotations encountered are labeled $A,B,C$, implementing $Z\!\to\!X$, $Z\!\to\!Y$, and $Y\!\to\!X$, and contributing $A_{ZX}B_{ZY}C_{YX}$ in Eq.~\eqref{eq:lgap_stag}.}
    \label{fig:stag_1f}
\end{figure}
The transformation of the surviving operators under one Floquet step is illustrated schematically in \figref{fig:stag_1f}.  
Yellow dots indicate the locations of nontrivial Pauli operators, while the absence of a dot denotes identity.  
As shown in \figref{fig:stag_1f},  the evolution involves three distinct single-qubit Haar unitaries, denoted \(A\), \(B\), and \(C\), which implement the Pauli rotations \(Z \to X\), \(Z \to Y\), and \(Y \to X\), respectively.

The Liouvillian gap in the staggered configuration can be computed via the same method used in the fully doped case.
It is given by
\begin{equation}
\label{eq:lgap_stag}
\begin{aligned}
    \Delta = 2\gamma
    - \frac{2}{N} \log \biggl(
        \left|
        \prod_{j=1}^{N/2}
        A^{j}_{ZX}B^{j}_{ZY}C^{j}_{YX}
        \right|
    \biggr),
\end{aligned}
\end{equation}
%C^{j-1} is equivalent to B^{j}
where \(A^j_{ZX}\), \(B^j_{ZY}\), and \(C^j_{YX}\) denote the matrix elements of the three Haar unitaries acting on \(j\) step.
Because our quantity is a logarithmic average, the dependence on \(\{A^j_{ZX}\}_{j=1}^{N/2}\), \(\{B^j_{ZY}\}_{j=1}^{N/2}\), and \(\{C^j_{YX}\}_{j=1}^{N/2}\) is immaterial and will be dropped.
In the large-\(N\) limit, Eq.~\eqref{eq:lgap_stag} simplifies to
\begin{align}
\label{eq:lgap_stag_1}
    \Delta = 
    2\gamma
    - \log\left| \exp\left(3\, \mathbb{E}_{U \sim \text{Haar}} \left[ \log |u| \right] \right) \right|,
\end{align}
where \(u\) denotes a representative non-diagonal Pauli-basis matrix element of a single-qubit Haar unitary. 
The form of Equation~\eqref{eq:lgap_stag_1} arises because all \(A^j\), \(B^j\), and \(C^j\) (for \(j = 1, \ldots, N/2\)) are independent Haar unitaries, yielding a total of \(3N/2\) independent random variables whose logarithms are averaged.
Substituting the Haar average of the logarithmic term, derived in Eq.~\eqref{eq:app_final}, into Eq.~\eqref{eq:lgap_stag_1}, we obtain
\begin{equation}
    \Delta = 2\gamma + 3.
\end{equation}

\section{Liouvillian Gap of Staggered-like Doping Structures at Weight-3 Truncation}
\label{app:struc_staggered_like}

\subsection{Existence of Cycle Formation}
\label{app:struc_staggered_like_cycle}

We show that cycles always appear in the support graph of the evolution operator, when projected onto the Pauli weight-3 subspace, for circuits with staggered-like doping configurations.
In these configurations, undoped qubits are separated by one or two doped sites.

As shown in Appendix~\ref{app:struc_staggered}, the staggered doping structure first allows non-nilpotency to emerge at Pauli weight \( w = 2 \), by enabling the transformation of operators such as \( Y_{2i}X_{2i+2} \to Y_{2i-2}X_{2i} \). As illustrated in \figref{fig:stag_1f}, this transformation is enabled when the local doping pattern takes the form
$\big[\!\bullet\!\circ\!\bullet\!\circ\!\;\big]$, where $\circ$ and $\bullet$ denote undoped and doped qubit lines, respectively.
Each dot in the pattern corresponds to the second through fifth qubit lines from the left in \figref{fig:stag_1f}.
The site labeled $X_{2i+2}$ before the transformation is not shown in this notation, as we are considering left-moving Pauli strings.

\begin{figure}
    \centering
    \includegraphics[width=0.98\linewidth]{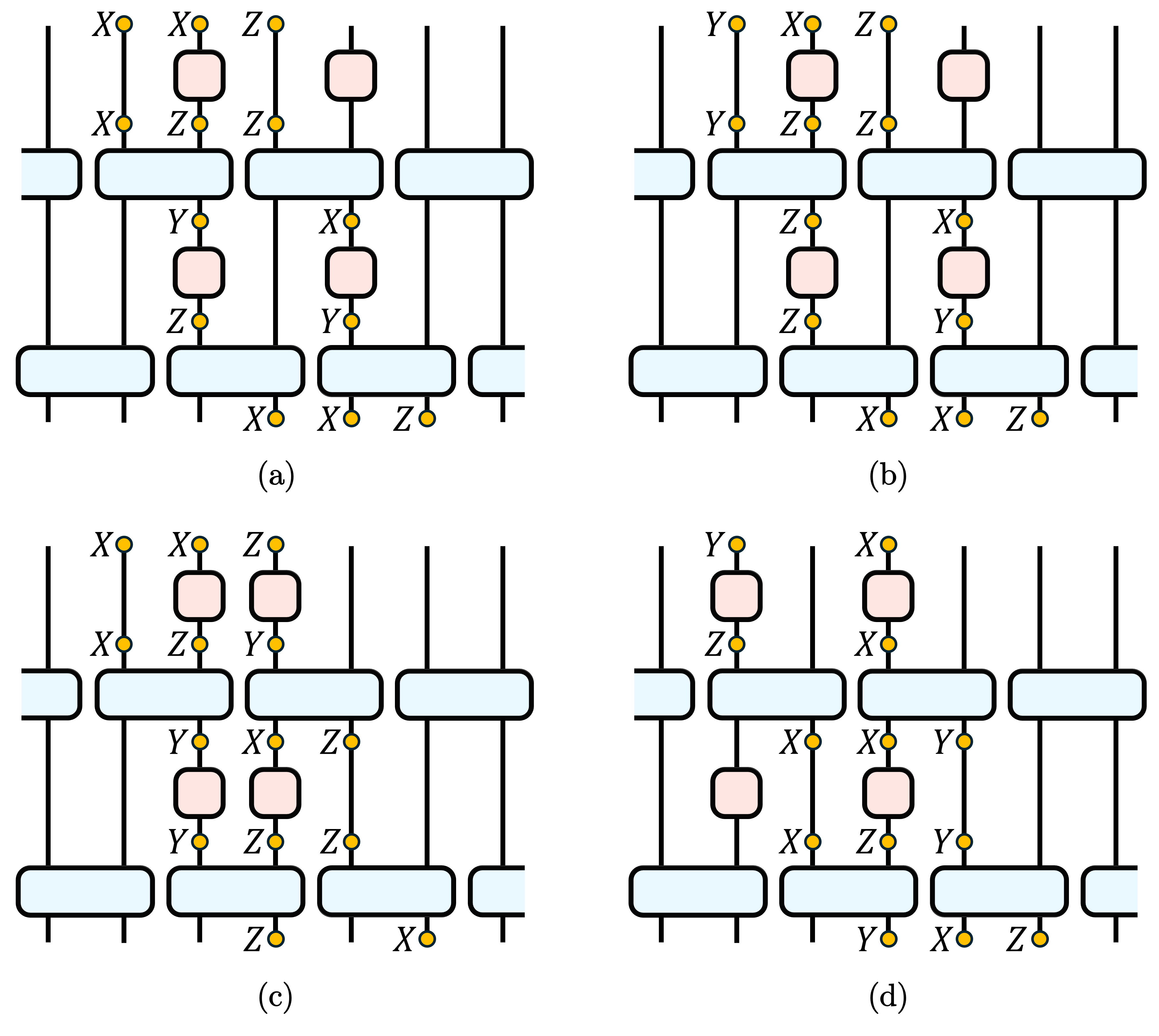}
    \caption{
    Additional one-step transformations used to build a weight-$3$ return cycle that occur in the staggered-like doping structure (cf.~\figref{fig:stag_1f}).
    Panels (a)–(d) show the required length-$3$ Pauli-string updates within the local patterns
$\big[\!\circ\!\bullet\!\circ\!\bullet\!\;\big]$, $\big[\!\circ\!\bullet\!\circ\!\bullet\!\;\big]$, $\big[\!\circ\!\bullet\!\bullet\!\circ\!\;\big]$, and $\big[\!\bullet\!\circ\!\bullet\!\circ\!\;\big]$, respectively, which realize the pathways in Eqs.~\eqref{eq:stag_like_path_1}--\eqref{eq:stag_like_path_3}.}

    \label{figs:staglike_1f}
\end{figure}

To fully establish the presence of a cycle at Pauli weight \( w = 3 \) for the staggered-like doping structure, four additional operator transformations are required. These transformations are illustrated in \figref{figs:staglike_1f}.
Figs.~\ref{fig:stag_1f} and~\ref{figs:staglike_1f} (d) both depict different operator transformations within the same circuit structure, $\big[\!\bullet\!\circ\!\bullet\!\circ\!\;\big]$. 
\figref{fig:stag_1f} shows the transformations \( Y_{2i}X_{2i+2} \to Y_{2i-2}X_{2i} \) and \( Y_{2i}X_{2i+2} \to Z_{2i-2}X_{2i} \); the latter corresponds to a variation in which the final Haar unitary applies a \( Z \to X \) transformation instead of \( Z \to Y \).
Similarly, \figref{figs:staglike_1f} (d) illustrates the transformation \( Y_{2i}X_{2i+1}Z_{2i+2} \to Y_{2i-2}X_{2i} \) and \( Y_{2i}X_{2i+1}Z_{2i+2} \to Z_{2i-2}X_{2i} \).
Figs.~\ref{figs:staglike_1f}(a) and (b) illustrate transformations within the circuit structure $\big[\!\circ\!\bullet\!\circ\!\bullet\!\;\big]$, corresponding to 
\( X_{2i}X_{2i+1}Z_{2i+2} \to X_{2i-2}X_{2i-1}Z_{2i} \) and 
\( X_{2i}X_{2i+1}Z_{2i+2} \to Y_{2i-2}X_{2i-1}Z_{2i} \), respectively.
Finally, \figref{figs:staglike_1f}(c) shows the transformation 
\( Z_{2i}X_{2i+2} \to X_{2i-2}X_{2i-1}Z_{2i} \), which arises from the circuit structure 
$\big[\!\bullet\!\circ\!\circ\!\bullet\!\;\big]$.

All of the transformations described above exhibit two key features. 
First, each Pauli string has length three. 
Second, the action consistently shifts both endpoints of the string two sites to the left.
These features imply that the Pauli string \( X_{2i} X_{2i+1} Z_{2i+2} \), which we denote as $XXZ$, can propagate leftward and return to its original form (up to translation) through three distinct pathways:
\begin{align}
&
XXZ
\xrightarrow{\circ\bullet\circ\bullet
,~\text{\figref{figs:staglike_1f} (b)}}
YXZ
\xrightarrow{\bullet\circ\bullet\circ
,~\text{\figref{figs:staglike_1f} (d)}}
ZIX
\xrightarrow{\circ\bullet\bullet\circ
,~\text{\figref{figs:staglike_1f} (c)}}
XXZ
,
\label{eq:stag_like_path_1}
\\
&
XXZ
\xrightarrow{\circ\bullet\circ\bullet
,~\text{\figref{figs:staglike_1f} (b)}}
YXZ
\xrightarrow{\bullet\circ\bullet\circ
,~\text{\figref{figs:staglike_1f} (d)}}
YIX
\rightarrow
\cdots
\rightarrow
YIX
\xrightarrow{\bullet\circ\bullet\circ
,~\text{\figref{fig:stag_1f}}}
ZIX
\xrightarrow{\circ\bullet\bullet\circ
,~\text{\figref{figs:staglike_1f} (c)}}
XXZ
,
\label{eq:stag_like_path_2}
\\
&
XXZ
\xrightarrow{\circ\bullet\circ\bullet
,~\text{\figref{figs:staglike_1f} (a)}}
XXZ
,
\label{eq:stag_like_path_3}
\end{align}
where in Eq.~\eqref{eq:stag_like_path_2}, the segment involving $YIX$ can repeat multiple times via the pattern 
$\big[\!\bullet\!\circ\!\bullet\!\circ\!\;\big]$.

When the sequence of transformations in Eq.~\eqref{eq:stag_like_path_1} is realized in the circuit, the corresponding doping structure and intermediate Pauli strings align vertically as
{
\setlength{\jot}{-2pt} 
\begin{align*}
&
\:\texttt{X}\,\texttt{X}\,\texttt{Z}
\\
&
\circ\bullet\bullet\circ 
\\
&
\quad\;\,
\:\texttt{Z}\,\texttt{I}\,\texttt{X}
\\
&
\quad\;\,
\bullet\circ\bullet\circ 
\\
&
\quad\;\,\quad\;\,
\:\texttt{Y}\,\texttt{X}\,\texttt{Z}
\\
&
\quad\;\,\quad\;\,
\circ\bullet\circ\bullet
\\
&
\quad\;\,\quad\;\,\quad\;\,
\:\texttt{X}\,\texttt{X}\,\texttt{Z}
\end{align*}
}
where the transformation sequence is displayed from bottom to top, matching the direction of operator propagation through the floquet steps.
Combining all segments, we obtain the following substructure of the circuit:
\begin{align}
\label{struc:stag_like_path_1}
\overline{\circ\bullet}
\bullet\circ\bullet\bullet
\underline{\circ\bullet}
\end{align}
Here, the underlined and overlined segments can be seamlessly connected, since both correspond to the leading \( XX \) portion of the $XXZ$ Pauli string.
Similarly, paths in Eqs.~\eqref{eq:stag_like_path_2} and~\eqref{eq:stag_like_path_3} correspond to the following substructures:
\begin{align}
\label{struc:stag_like_path_2}
\overline{\circ\bullet}
\bullet\circ
(\bullet\circ)^k
\bullet\bullet
\underline{\circ\bullet}
\end{align}
and
\begin{align}
\label{struc:stag_like_path_3}
\overline{\circ\bullet}
\underline{\circ\bullet}
,
\end{align}
respectively. Here, \( k \) denotes the number of repetitions of the pattern \( (\bullet\circ) \).

We conclude by noting that all staggered-like configurations can be generated by concatenating the substructures in Eqs.~\eqref{struc:stag_like_path_1}-\eqref{struc:stag_like_path_3}.
To make use of the structural features identified above, we represent all staggered-like configurations in the following form:
\begin{align}
\label{struc:stag_like_gen}
\overline{\circ\bullet}
(\circ\bullet)^{k_0-1}
(\circ\bullet\bullet)^{k_1}
(\circ\bullet)^{k_2}
(\circ\bullet\bullet)^{k_3}
\cdots
(\circ\bullet)^{k_{l-1}}
(\circ\bullet\bullet)^{k_l}
\underline{\circ\bullet},
\end{align}
where the overlined segment indicates the first \(\circ\bullet\) pair when \( k_0 \neq 0 \).
All other \( k_i \) are assumed to be nonzero; otherwise, adjacent identical patterns can be merged and the expression simplified.

If the pattern \((\circ\bullet)\) appears at the end of the sequence, it can be cyclically shifted and absorbed into a nonzero \(k_0\), due to the periodic boundary condition. In contrast, if the configuration contains no \((\circ\bullet)\) segments at all, it corresponds to the special case \(l = 1\), and is represented as
\begin{align}
\label{struc:stag_like_spec}
\overline{\circ\bullet}\bullet
(\circ\bullet\bullet)^{k_1-1} 
\underline{\circ\bullet}.
\end{align}
Although this form appears exceptional, the structure in Eq.~\eqref{struc:stag_like_spec} still corresponds to \( k_1/2 \) repetitions of the substructure in Eq.~\eqref{struc:stag_like_path_1}.

We now return to the general form in Eq.~\eqref{struc:stag_like_gen} and describe how it can be constructed from left to right using the elementary substructures.
The segment \((\circ\bullet)^{k_0}\) is realized by concatenating the substructure in Eq.~\eqref{struc:stag_like_path_3} exactly \(k_0\) times. The segment \((\circ\bullet\bullet)^{k_1}\) depends on the parity of \(k_1\): if \(k_1\) is even, it can be built from \(k_1/2\) repetitions of the substructure in Eq.~\eqref{struc:stag_like_path_1}; if \(k_1\) is odd, one uses \(\lfloor k_1/2 \rfloor\) copies of Eq.~\eqref{struc:stag_like_path_1} followed by a single copy of Eq.~\eqref{struc:stag_like_path_2}.
In the latter case, the $YIX$ segment introduced in Eq.~\eqref{struc:stag_like_path_2} must be repeated \(k_2\) times to account for the \((\circ\bullet)^{k_2}\) portion. If \(k_1\) was even, this \((\circ\bullet)^{k_2}\) segment is instead constructed by repeating the substructure in Eq.~\eqref{struc:stag_like_path_3} \(k_2\) times.

Continuing in this manner---alternating between \((\circ\bullet)\) and \((\circ\bullet\bullet)\) segments---and choosing the appropriate combinations based on the parity of each \(k_i\), one can construct any general configuration of the form given in Eq.~\eqref{struc:stag_like_gen}.
The fact that all staggered-like configurations can be constructed from these substructures implies that there always exists a cycle in the operator propagation: the Pauli string starting as $XXZ$ eventually returns to the same operator at the same sites. Moreover, since each substructure involves only transformations between Pauli strings of weight two or three, we conclude that cycles always appear in the support graph of the evolution operator, when projected onto the Pauli weight-3 subspace.

\subsection{Upper Bound of Liouvillian Gap}
\label{app:Lgap_upp_bound}

To estimate the spectral radius $\rho\!\left(\tilde{\Phi}^{\mathrm{trunc}}_F\right)$ of the truncated Floquet operator projected onto the Pauli weight-$\leq 3$ subspace, we recall Gelfand’s formula,
\begin{equation}
\rho\!\left(\tilde{\Phi}^{\mathrm{trunc}}_F\right)
= \lim_{n\to\infty} \bigl\| \bigl( \tilde{\Phi}^{\mathrm{trunc}}_F \bigr)^n \bigr\|^{1/n}.
\end{equation}
In particular, the spectral radius can be lower bounded by the growth rate of any nonvanishing matrix element:
\begin{equation}
\rho\!\left(\tilde{\Phi}^{\mathrm{trunc}}_F\right)
\ge
\limsup_{n\to\infty}
\left| \mathrm{Tr}\!\left[ P\, (\tilde{\Phi}^{\mathrm{trunc}}_F)^n(P) \right] \right|^{1/n},
\end{equation}
for any Pauli string $P$ with $\|P\|_2 = 1$. If the overlap accumulates coherently under iteration, that is, if
\begin{equation}
\mathrm{Tr}\!\left[P\, \bigl(\tilde{\Phi}^{\mathrm{trunc}}_F\bigr)^{mL}(P)\right] = c^m
\quad \text{for some } c \neq 0,
\end{equation}
then the bound simplifies to
\begin{equation}
\label{eq:specR_low_bound}
\rho\!\left(\tilde{\Phi}^{\mathrm{trunc}}_F\right)
\ge
\left| \mathrm{Tr}\!\left[
P\, \bigl( \tilde{\Phi}^{\mathrm{trunc}}_F \bigr)^L (P)
\right] \right|^{1/L}.
\end{equation}

In Appendix~\ref{app:struc_staggered_like_cycle}, we constructed such a Pauli operator $P$ explicitly within a class of staggered-like doping patterns, and proved that for any system size~$N$, the truncated evolution admits a length-$L$ cycle with $L = N/2$ such that this multiplicative structure holds exactly. While this establishes the existence of at least one such return path, it remains to determine whether the length-$N/2$ cycle is unique or whether multiple disjoint cycles contribute. We address this in the remainder of this subsection.

Let us fix as a reference the return cycle identified in Appendix~\ref{app:struc_staggered_like_cycle}, generated by the Pauli string $P = YXZ$, supported on a local configuration of the form $\big[\!\circ\!\bullet\!\bullet\!\circ\!\;\big]$.
We take this operator as a representative reference in our analysis. Even if the circuit deviates from the ideal brickwork arrangement, the same structure can be restored by conjugating the unitary block with Hadamard layers on both ends and taking its Hermitian adjoint. Since this modification involves only single-qubit transformations, it preserves the eigenvalue under the noise model. If no such configuration appears in the circuit, the system falls into the staggered regime, for which return cycles can be constructed using the methods of Appendix~\ref{app:struc_staggered}.

Under repeated Floquet evolution, the leftmost $Y$ propagates two sites to the left at each step. If the rightmost $Z$ does not shift by two sites to the left as well---for instance, if it moves one site to the right---then it proceeds to propagate two sites to the right at each subsequent Floquet step. This causes the left and right endpoints to diverge in opposite directions. Once the support crosses into undoped regions, the weight increases and the component becomes truncated. Therefore, any Pauli string that remains invariant under projection must be confined within three consecutive sites.

\begin{table}
    \centering
    \renewcommand{\arraystretch}{1.1}
    \begin{tabular}{c@{\quad}c@{\quad}c@{\quad}c}
         \hline
         \hline
         No. & Structure & Input strings & Output strings\\
         (i) &
         $\big[\!\circ\!\bullet\!\circ\!\bullet\!\;\big]$&
         $XXZ,XXY,XIX$& 
         $XXZ,YXZ$\\
         (i$'$) &
         $\big[\!\bullet\!\bullet\!\circ\!\bullet\!\;\big]$&
         $XXZ,XXY,XIX$&
         $XXZ,YXZ,ZXZ$\\
         (ii) &
         $\big[\!\bullet\!\circ\!\bullet\!\circ\!\;\big]$&
         $(YXZ,YIX)^*,(XXY,ZXZ,ZIX)^{**}$&
         $(PIP')^*,(PZP')^{**}$\\
         (ii$'$) &
         $\big[\!\bullet\!\circ\!\bullet\!\bullet\!\;\big]$&
         $(YXZ,YXY,YIX)^*,(XXZ,XXY,XIX,ZXZ,ZXY,ZIX)^{**}$&
         $(PIP')^*,(PZP')^{**}$\\
         (iii) &
         $\big[\!\circ\!\bullet\!\bullet\!\circ\!\;\big]$&
         $XXY,YXZ,YIX,ZXZ,ZIX$&
         $ZIP,XPP',YPP'$\\
         \hline
         \hline
    \end{tabular}
    \caption{Local Pauli transformations under one Floquet step for each four-site structure, assuming the three-site support constraint. 
    Here, $P, P'$ denote arbitrary nontrivial single-qubit Pauli operators ($X$, $Y$, or $Z$).
    Grouped expressions marked by $(\cdots)^*$ or $(\cdots)^{**}$ indicate that only transformations within the same group are allowed.}
    \label{tab:local_trans}
\end{table}

Under the three-site support constraint, all possible Pauli string transitions induced by a single Floquet step have been enumerated for every four-site local configuration. These are summarized in Table~\ref{tab:local_trans}.
However, not all of the input strings listed in Table~\ref{tab:local_trans} remain within the weight-$\leq 3$ subspace under repeated Floquet evolution. Whether a given Pauli string survives without truncation depends on the global circuit structure. This behavior is illustrated in \figref{fig:local_trans}.
As shown in \figref{fig:local_trans}, the range of transformations extends beyond those captured in Appendix~\ref{app:struc_staggered_like_cycle}, namely Eqs.~\eqref{eq:stag_like_path_1}--\eqref{eq:stag_like_path_3}. This broader connectivity arises from the inclusion of additional strings in the transitions: for example, the output $ZIX$ in case~(ii$'$), the input $YIX$ in case~(iii), and multiple appearances of $XXY$ as both input and output in cases~(i), (i$'$), and~(iii). Taking these into account, the list of possible return cycles can be refined as follows:
\begin{align}
&YXZ
\to
YIX
\xrightarrow{k\text{\;times}}
YIX
\to
\left\lgroup
\begin{array}{l}
YIX\\
ZIX\\
\end{array}
\right\rgroup
\to
\left\lgroup
\begin{array}{l}
XXY\\
XXZ\\
\end{array}
\right\rgroup
\to
YXZ,
\label{eq:stag_like_path_1_new}
\\
&YXZ
\to
YIX
\xrightarrow{k\text{\;times}}
YIX
\to
\left\lgroup
\begin{array}{l}
YIX\\
ZIX\\
\end{array}
\right\rgroup
\to
\left\lgroup
\begin{array}{l}
XXY\\
XXZ\\
\end{array}
\right\rgroup
\to
XXZ
\xrightarrow{k'\text{\;times}}
XXZ
\to
YXZ,
\label{eq:stag_like_path_2_new}
\end{align}
where $k,k' \in \mathbb{Z}_{\ge 0}$ denote the number of identity-preserving steps. If either $k$ or $k'$ equals zero, the corresponding transformation step is omitted.

\begin{figure}
    \centering
    \includegraphics[width=0.75\linewidth]{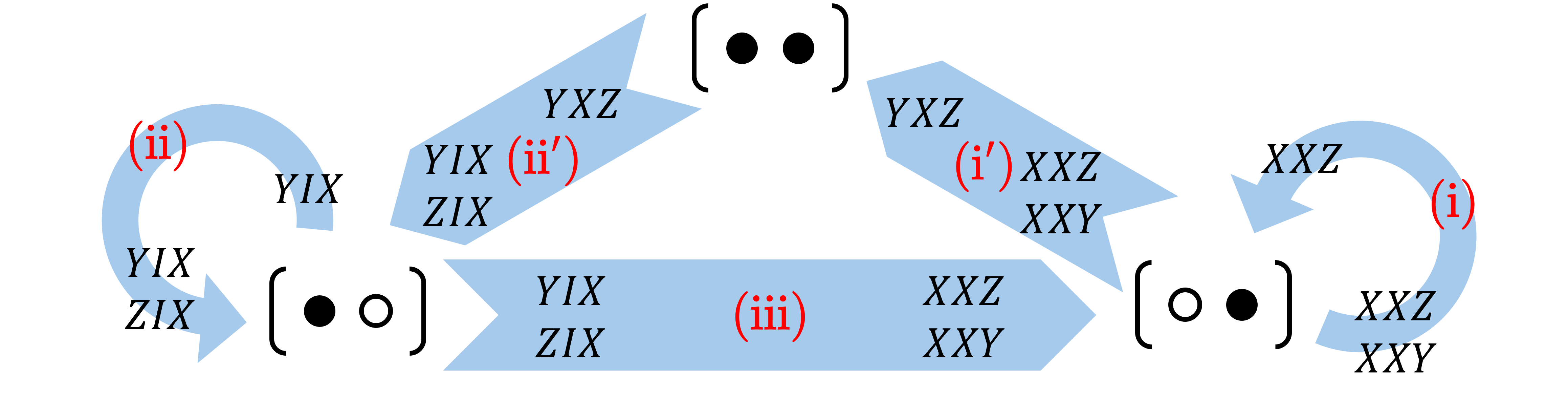}
    \caption{
    Diagram of transformation pathways for contiguous three-site Pauli strings that survive under repeated Floquet evolution. 
    Each arrow corresponds to a four-site substructure, constructed by joining a two-site core substructure on the left (from the set $\big[\!\bullet\!\bullet\!\;\big]$, $\big[\!\circ\!\bullet\!\;\big]$, or $\big[\!\bullet\!\circ\!\;\big]$) with one on the right, as indicated by the arrow direction (left to right). 
    These two-site patterns refer to the mid-two sites (positions 2 and 3) of the four-site configurations listed in Table~\ref{tab:local_trans}. 
    The Pauli operator at the tail of the arrow is mapped to the operator at the head through a single Floquet step. The red label on each arrow indicates the rule number in Table~\ref{tab:local_trans} used for that transformation.
}
    \label{fig:local_trans}
\end{figure}

We now compute the correlation function 
\[
\mathrm{Tr}
    \left[
    (YXZ_i)
    \,
    \tilde{\Phi}^{\mathrm{trunc}}_F\left(YXZ_{i+2n}\right)
    \right],
\]
where \( i \) denotes the starting site of the contiguous three-site Pauli string, and \( n \) is the total number of Floquet steps required for the string to return to a shifted version of itself, following the transformation paths described in Eqs.~\eqref{eq:stag_like_path_1_new} or~\eqref{eq:stag_like_path_2_new}, depending on the specific local substructure.
This correlation function can be represented as
\begin{align}
\label{eq:sub_lower_bound}
    \mathrm{Tr}
    \left[
    (YXZ_i)
    \,
    \tilde{\Phi}^{\mathrm{trunc}}_F\left(YXZ_{i+2n}\right)
    \right]
    =
    (A_{ZZ}B_{YY}+A_{ZY}B_{XY})
    (A_{YZ}B_{YX}+A_{YY}B_{XX})
    \prod_{i=1}^{m(n)}
    C_{i}
    .
\end{align}
We define a set \(\{C_i\}_{i=1}^{m(n)}\) of single-qubit Haar-random unitaries that appear at fixed path-independent positions within the cycle. Each \(C_i\) denotes a Pauli-basis coefficient, and the total number \(m(n)\) counts such positions once per Floquet step, even if they reappear along distinct paths. Although correlations may exist among the \(C_i\), they become negligible when taking the logarithmic Haar average, which is the focus of our analysis.  
In contrast, the coefficients \(A_{ij}\) and \(B_{ij}\) represent Haar-random unitaries associated with path-dependent segments that differentiate distinct return paths to the operator \(YXZ\), as illustrated in Eqs.~\eqref{eq:stag_like_path_1_new} and \eqref{eq:stag_like_path_2_new}.

We are now ready to apply the general spectral radius lower bound given in Eq.~\eqref{eq:specR_low_bound}, with \( P = YXZ \) and \( L = N/2 \). The trace term
\(
\mathrm{Tr}\!\left[
P\, \bigl( \tilde{\Phi}^{\mathrm{trunc}}_F \bigr)^L (P)
\right]
\)
can be expressed as a product of correlation functions of the form derived in Eq.~\eqref{eq:sub_lower_bound}, stitched together according to the return paths defined by the substructure.
As a result, the spectral radius admits the following lower bound:
\begin{align}
    \rho\!\left(\tilde{\Phi}^{\mathrm{trunc}}_F\right)
    \ge
    e^{-3\gamma}
    \left|
    \prod_{i=1}^{\alpha(L)}
    \left[
    (A_{ZZ;i}B_{YY;i}+A_{ZY;i}B_{XY;i})
    (A_{YZ;i}B_{YX;i}+A_{YY;i}B_{XX;i})
    \right]
    \prod_{j=1}^{\beta(L)}
    C_{j}
    \right|^{1/L},
\end{align}
where we have uniformly lower bounded each dissipation factor per Floquet step by \(e^{-3\gamma}\). The exponents \(\alpha(L)\) and \(\beta(L)\) count the number of contributing transformations and are at most proportional to \(L\), depending on the global structure of the Floquet operator.
Accordingly, the Liouvillian gap is upper bounded as
\begin{align}
    \Delta_{\text{stag}^*}
    \le 
    3\gamma
    + c,
\end{align}
where
\begin{align}
    c &=
    -\frac{1}{L} \sum_{i=1}^{\alpha(L)}
    \left[
    \log|A_{ZZ;i}B_{YY;i}+A_{ZY;i}B_{XY;i}|
    +
    \log|A_{YZ;i}B_{YX;i}+A_{YY;i}B_{XX;i}|
    \right]
    -
    \frac{1}{L} \sum_{i=1}^{\beta(L)}
    \log|C_i|.
\end{align}
Assuming $\alpha(L)$ and $\beta(L)$ are sufficiently large to allow averaging, we approximate
\begin{align}
    c \approx \bar{\alpha}(4 - 2\log 2) + \bar{\beta},
\end{align}
with $\bar{\alpha} := \alpha(L)/L$ and $\bar{\beta} := \beta(L)/L$. A detailed derivation of the Haar average is provided in Appendix~\ref{app:haar-average}.

Finally, we can bound the constant $c$ uniformly by analyzing which return paths maximize the averaged logarithmic factors.
The largest value is attained when the cycle spends the maximal fraction of steps on the repeating $\big[\;\!\bullet\;\!\circ\!\;\big]$ pattern where the string $YIX$ loops back to itself.
In this case, one finds $c=3$.
Note, however, that this configuration belongs to the staggered regime, for which the surviving string has weight~2 throughout the cycle; using the corresponding dissipation factor gives the tighter bound $\Delta=2\gamma+3$. For circuits in the staggered-like class $\mathrm{stag}^*$ where weight-3 segments necessarily appear, the largest slope-compatible contribution arises from repeating the three-step motif
(ii$'$)$\to$(iii)$\to$(i$'$) in \figref{fig:local_trans}.
This yields 
\begin{align}
c=\frac{10-2\ln 2}{3}\approx 2.9 \leq 3.
\end{align}
Therefore, in all cases we have the uniform bound $c\leq 3$, and hence
\begin{equation}
\Delta_{\mathrm{stag}^*}\le 3\gamma+3.
\end{equation}

\section{Cycles in Block-Staggered Doping Structures}
\label{app:struc_block_staggered}

\begin{table}[t]
    \centering
    \begin{ruledtabular}
    \begin{tabular}{c c l}
        $k$ & $w_\ast$ & Representative return cycle 
    \\[1pt]
    \hline
    \\[0pt]
    2
    &
    5
    & 
    $
    XXYXY
    \to 
    YZYXZ
    \to 
    ZIZIX
    $
    \\[1pt]
    3 
    &
    5
    &  
    $YZYXZ\to ZIZIX$
    \\[1pt]
    4    
    &
    7
    & 
    $
    XZIYXIZ
    \to
    YZXXZXY
    \to
    ZIZYIXZ
    \to
    XZZIXXY
    \to
    YZZIZXZ
    $
    \\[1pt]
    5   
    &
    8
    & 
    $
    XZYYZYXIX
    \to
    YZXZIXZXY
    \to
    ZIZYXYYXZ
    $
    \\[1pt]
    6     
    &
    7
    &  
    $
    YXXXIIXXY
    \to
    ZIIXZYXXZ
    \to
    XZZYXYIIX
    \to
    ZZZIIZZXY
    $
    \\
    &&
    $
    \to
    XZXZIZZIX
    \to
    YZIIYIIXY
    \to
    ZIXXIXXXZ
    $
    \\[1pt]
    7    
    &
    9
    &  
    $
    YIZYYYXIX
    \to
    ZIZIYXZXY
    \to
    XZYXXZYXZ
    \to
    YZXZYIZIX
    $
    \\[1pt]
    8     
    &
    11
    &  
    $
    ZIIXXIXXIYIIX
    \to
    XZZYXXXXIIXXY
    \to
    YZZIIZZXZYXXZ
    $
    \\
    &&
    $
    \to
    ZIIYIZZIZYIIX
    \to
    XZXXYZZZXIIXY
    \to
    YZIIYZZIIZZXZ
    $
    \\
    &&
    $
    \to
    ZIXXIXXZIZZIX
    \to
    XZXXIIXXXIIXY
    \to
    YYIIZXXXYZXXZ
    $
    \\
    \end{tabular}
    \end{ruledtabular}
    \caption{
        Representative return cycles of period \(L = N/2\) for minimal substructures containing \(k\) undoped qubits.
        Each cycle begins and ends at the same Pauli string after propagating through a repeating block-staggered configuration.
        For odd (even) \(k\), the cycle propagates \(k+1\) (\(2k+2\)) sites before returning to the original Pauli string.
        Here, each representative cycle is displayed in a convention where the initial Pauli string is right-aligned with a doped site, such that its rightmost nontrivial Pauli operator acts on a doped qubit.
    }
    \label{tab:block_staggered_cycles}
\end{table}

Table~\ref{tab:block_staggered_cycles} lists, for each block length \(k\), the minimal truncation weight \(w_\ast\) at which our construction yields a localized return cycle in the block-staggered pattern.
For each value of \(k\), only one representative example is shown, although multiple distinct return cycles may exist.
Using the procedure described in the main text, all such cycles can be generated efficiently by enumerating the allowed local transitions within the truncated subspace.
Each example is defined on a minimal repeating substructure; when this substructure is tiled to form the full circuit, the corresponding Pauli string returns to the same operator at the original position after \(L=N/2\) Floquet steps.
As in the staggered-like case, the existence of these localized return cycles is sufficient to fix the \(N\)-independent upper bound on the Liouvillian gap in the strong dissipation limit.

\end{document}